\documentclass[11pt]{article}

\usepackage[utf8]{inputenc} 
\usepackage[T1]{fontenc}    
\usepackage{url}            
\usepackage{nicefrac}       

\usepackage[margin = 1 in]{geometry}
\usepackage{amsmath, amssymb, amsthm, thm-restate}
\usepackage{bm}
\usepackage{braket, physics, enumitem, relsize}
\usepackage{makecell, graphicx, subfiles, multirow, wrapfig}
\usepackage{tikz}
\usepackage[font=scriptsize]{subcaption}
\usepackage[font=scriptsize, justification=raggedright]{caption}
\usepackage[colorlinks, allcolors=black, pdfencoding=auto, psdextra]{hyperref}
\usepackage[capitalize, nameinlink, noabbrev]{cleveref}
\usepackage{tabularx, booktabs, array}
\setlist[itemize]{leftmargin=10pt}
\usepackage{float}
\usepackage{neuralnetwork}
\usepackage{adjustbox}
\usepackage{comment}
\usepackage{tikz}
\usetikzlibrary{shapes.geometric, positioning, trees}

\usepackage{lmodern}

\usetikzlibrary{quantikz2}
\usepackage{xcolor}
\hypersetup{
  colorlinks   = true, 
  urlcolor     = blue, 
  linkcolor    = teal, 
  citecolor   = blue 
}
\usepackage{setspace}
\usepackage{algcompatible, newfloat}
\usepackage{mathpazo}
\DeclareFloatingEnvironment[
    fileext=loa,
    listname=List of Algorithms,
    name=Algorithm,
    placement=t,
]{alg}
\usepackage{algpseudocode}
\usepackage{algorithm}
\usepackage{bm}

\mathcode`l="8000
\begingroup
\makeatletter
\lccode`\~=`\l
\DeclareMathSymbol{\lsb@l}{\mathalpha}{letters}{`l}
\lowercase{\gdef~{\ifnum\the\mathgroup=\m@ne \ell \else \lsb@l \fi}}%
\endgroup

\newtheorem{lemma}{Lemma}
\newtheorem{definition}{Definition}

\newtheorem{corollary}{Corollary}

\newtheorem{theorem}{Theorem}

\newtheorem*{theorem*}{Theorem}

\newtheorem{claim}{Claim}
\newtheorem{fact}{Fact}

\numberwithin{lemma}{section}
\numberwithin{definition}{section}
\numberwithin{remark}{section}
\numberwithin{corollary}{section}
\numberwithin{proposition}{section}
\numberwithin{theorem}{section}
\numberwithin{example}{section}

\Crefname{proposition}{Proposition}{Propositions}
\Crefname{alg-line}{Line}{Lines}
\crefname{alg}{Algorithm}{Algorithms}

\newcommand{\wh}[1]{\widehat{#1}}

\DeclareMathOperator{\Var}{Var}

\renewcommand{\epsilon}{\varepsilon}

\newcolumntype{x}[1]{>{\centering\hspace{0pt}\arraybackslash}m{#1}}

\usepackage{amsmath}
\usepackage{graphicx}
\usepackage{amsfonts}
\usepackage{float}
\usepackage{amssymb}

\newcommand{\blankline}{\quad\pagebreak[2]}

\DeclareMathOperator{\Ima}{Im}
\DeclareMathOperator{\E}{\mathrm{E}}

\title{The Space Just Above One Clean Qubit}
\author{Dale Jacobs \footnote{Tufts CS, Medford MA, USA, dale.jacobs@tufts.edu} \quad Saeed Mehraban \footnote{Tufts CS, Medford MA, USA, saeed.mehraban@tufts.edu}}
\date{\today}

\begin{document}

\maketitle
\begin{abstract}
    Consider the model of computation where we start with two halves of a $2n$-qubit maximally entangled state. We get to apply a universal quantum computation on one half, measure both halves at the end, and perform classical postprocessing. This model, which we call $\frac12$\textbf{BQP}, was defined in STOC 2017 \cite{aaronson2017computational} to capture the power of permutational computations on special input states. As observed in \cite{aaronson2017computational}, this model can be viewed as a natural generalization of the one-clean-qubit model (\textbf{DQC}1) where we learn the content of a high entropy input state only after the computation is completed. An interesting open question is to characterize the power of this model, which seems to sit nontrivially between \textbf{DQC}1 and \textbf{BQP}. In this paper, we show that despite its limitations, this model can carry out many well-known quantum computations that are candidates for exponential speed-up over classical (or possibly \textbf{DQC}1) computations. In particular, $\frac12$\textbf{BQP} can simulate Instantaneous Quantum Polynomial Time (\textbf{IQP}) and solve the Deutsch-Jozsa problem, Bernstein-Vazirani problem, Simon's problem, and period finding. As a consequence, $\frac12$\textbf{BQP} also solves \textsc{Order Finding} and \textsc{Factoring} outside of the oracle stting. Furthermore, $\frac12$\textbf{BQP} can solve \textsc{Forrelation} and the corresponding oracle problem given by Raz and Tal \cite{raz-tal} to separate \textbf{BQP} and \textbf{PH}. We also study limitations of $\frac12$\textbf{BQP} and show that similarly to \textbf{DQC}1, $\frac12$\textbf{BQP} cannot distinguish between unitaries which are close in trace distance, then give an oracle separating $\frac12$\textbf{BQP} and \textbf{BQP}. Due to this limitation, $\frac12$\textbf{BQP} cannot obtain the quadratic speedup for unstructured search given by Grover's algorithm \cite{grover1996fast}. We conjecture that $\frac12$\textbf{BQP} cannot solve $3$-\textsc{Forrelation}.
\end{abstract}


\tableofcontents

\section{Introduction}
\label{sec:introduction}
One of the central goals of quantum complexity theory is understanding the relationships between quantum and classical models of computation. In particular, a major open problem in the field is separating the complexity classes \textbf{BPP} and \textbf{BQP}. Informally, this problem can be formulated as, ``(how) do quantum computers provide a computational advantage over classical computers?''
While progress has been made toward this goal, whether or not $\textbf{BPP} {\neq} \textbf{BQP}$ remains wide open, as an unconditional separation would resolve outstanding unsolved conjectures, including $\textbf{P} {\neq} \textbf{PSPACE}$. Due to the usual difficulties in proving such separations and motivated by near-term quantum advantage experiments, a fruitful line of research has been classifying intermediate models of quantum computation, which are more powerful than classical computation but not universal for quantum computation. Important examples of such models are the one-clean-qubit model (\textbf{DQC}1) \cite{Knill_1998}, instantaneous quantum computation (\textbf{IQP}) \cite{Shepherd_2009,shepherd2010quantum, bremner2011classical}, linear optics \cite{aaronson2011computational}, low-depth quantum circuits (\textbf{QAC}, \textbf{QNC}, etc) \cite{moore1999quantum,hoyer2003quantum}, and intermediate-scale quantum circuits with noise (\textbf{NISQ}) \cite{preskill2018quantum, chen2023complexity, chia2024oracle}. Studying these sub-universal models provides insight into the resource trade-offs surrounding the \textbf{BPP} vs \textbf{BQP} question. 
 
A standard formula for defining a sub-universal quantum model is to begin with a \textbf{BQP}-universal model and impose restrictions that make the model weaker. Usually, these models are also physically motivated (e.g., linear optical circuits \cite{aaronson2011computational} or permutational circuits in \cite{aaronson2017computational}). Importantly, many of the well-studied sub-universal quantum models seem incomparable due to the large degree of freedom in considering which resources to restrict. For example, the one-clean-qubit model (\textbf{DQC}1), which captures the power of computation on one bit of quantum information and many quantum bits in the maximally mixed state \cite{Knill_1998}, provides an advantage over classical computation by estimating (normalized) traces of quantum circuits. In contrast, the power of Instantaneous Quantum Polynomial Time (\textbf{IQP}) \cite{Shepherd_2009, Bremner_2010}, which models the power of commuting Hamiltonians that can be applied instantly, seems to stem from its ability to perform certain classical computations between two layers of Hadamard gates. Thus, a meaningful step toward understanding the boundary between classical and quantum computation would be to develop a chain of inclusions consisting of increasingly powerful comparable sub-universal models of quantum computation. Within this framework, we can hope to achieve a more fine-grained perspective on when and how quantum resources provide a computational advantage.
 
Inspired by these motivations, the main focus of this work is to study the following generalization of the one-clean-qubit model: We start with $n$ qubits initially in the maximally mixed state (zero clean qubits!). We get to apply a polynomial-size quantum circuit from a universal gate set. We then measure all qubits, and in addition, we will learn the initial state from which we started. We will then feed the bit string we sampled and the content of the initial state to (\textbf{BPP}) classical postprocessing. This model is a generalization of \textbf{DQC}1 because it can simulate one clean qubit by assuming that the first bit is in the pure $\ket{0}$ state, performing a \textbf{DQC}1 computation, then verifying this assumption in the end. One can equivalently define this model based on computation over one-half of a maximally entangled EPR state. In the end, we measure both halves and proceed to classical postprocessing. The content of the second half can be viewed as the initial state the model started from, and we will learn it only in the end. Because of this analogy, we call the class of computational problems solvable by this model $\frac12$\textbf{BQP}. This model was defined briefly in \cite{aaronson2017computational} (under the name \textbf{SampTQP}, standing for Trace Quantum Polynomial-Time) and has received little attention since. 

There are several reasons for studying this model. First, as discussed above, this model is a natural generalization of one clean qubit. Understanding the full capabilities of this model would delineate the boundary between \textbf{DQC}1 and \textbf{BQP}. Second, beyond the complexity-theoretic goal of better understanding quantum resources, sub-universal quantum models are often more feasible to instantiate in a lab and, therefore, provide intermediate steps toward practical quantum computation. From this perspective, $\frac12$\textbf{BQP} represents loading a quantum state into a quantum register via a noisy process and later learning the specific noise that occurred, perhaps via thermal imaging. In this setting, $\frac12$\textbf{BQP} represents the scenario in which loading the quantum state is very noisy, but the remainder of the circuit is not subjected to noise.
From a broader perspective, since current state-of-the-art quantum advantage experiments usually perform tasks that are contrived to be classically difficult, finding sub-universal models that can perform more natural computational tasks brings us closer to realizing practical quantum advantage. In this sense, identifying intermediate quantum models stronger than known sub-universal models is practically justified in its own right. Finally, as motivated in \cite{aaronson2017computational}, this model captures the power of restricted gate sets with many conserved quantities, such as exchange interactions (ball permutations), on specific inputs and output measurement basis.

Due to the restrictions of $\frac12$\textbf{BQP}, we expect it would be subject to similar computational limitations as \textbf{DQC}1. At the same time, because the model learns more information than \textbf{DQC}1 at the end of the computation, we expect it to be more powerful. Where between \textbf{DQC}1 and \textbf{BQP} does the power of $\frac12$\textbf{BQP} rest? In this work, we address this question and show that not only can $\frac12$\textbf{BQP} simulate \textbf{DQC}1, it can furthermore simulate other models such as \textbf{IQP}, and captures many known results separating classical and quantum computation, in particular, $\frac12$\textbf{BQP} can solve the following problems: Deutsch-Jozsa, Bernstein-Vazirani, Simon's \cite{simon1997power}, \textsc{Period Finding}, \textsc{Order Finding}, and \textsc{Factoring}. We furthermore show that this model can solve \textsc{Forrelation} 
\cite{aaronson2014forrelation}, and the Raz-Tal problem \cite{raz-tal}, which imply an oracle separation with the polynolmial Hierarchy. On the other hand, we show that $\frac{1}{2}$\textbf{BQP} does not provide speedup for certain problems including unstructured search.

An intriguing question in quantum computing is whether we can define a sub-universal model encompassing all known quantum algorithmic primitives. Our work demonstrates that the power of \(\frac{1}{2}\)\textbf{BQP}, as a rudimentary model of quantum computation, is sufficient to capture many well-known quantum algorithms, but not all. This provides valuable insight into the resource-theoretic nature of quantum algorithms.

\subsection{The model}
\label{subsec:the-model}

The $\frac12$\textbf{BQP} model of computation was first described under the name ``\textbf{SampTQP}'' in \cite{aaronson2017computational}, motivated by a representation-theoretic analysis of the $1+1$-dimensional quantum ball permuting model of computation (which we will explain in detail in \cref{subsec:ball-perm}). 
The $\frac12$\textbf{BQP} model, however, is interesting in its own right as a sub-universal model of quantum computation which generalizes specific well-known candidates for demonstrating quantum advantage such as \cite{Knill_1998, Shepherd_2009}.
In particular, $\frac12$\textbf{BQP} is an interesting toy model highlighting certain resource trade-offs between sub-universal quantum models and \textbf{BQP}.
 
In the $\frac12$\textbf{BQP} model, we start with the input state:
\begin{align*}
    \frac{1}{\sqrt{2^n}}\sum_{x \in \{0,1\}^n}\ket{x}\ket{x}
\end{align*}
on $2n$ qubits, where the first $n$ qubits are maximally entangled to the remaining $n$ qubits. We are allowed to apply an arbitrary quantum circuit to the first $n$ qubits only, and then measure all $2n$ qubits in the computational basis. After measurement, we allow \textbf{BPP} post-processing.
 
Equivalently, one can think of a $\frac12$\textbf{BQP} machine as first sampling from a quantum circuit $C$ whose input is a uniformly random computational basis state $\ket{w}$, then learning $w$ at the end of the computation and performing \textbf{BPP} post-processing. This equivalence is made explicit in \cref{subsec:random-bqp}.
 
The $\frac12$\textbf{BQP} complexity class is the set of problems solvable by a $\frac12$\textbf{BQP} machine as described above. Our goal is to understand the computational power of this model. 

\subsection{Our contributions}
\label{subsec:our-contributions}

We now present our main results. Our main goal is to bring the $\frac12$\textbf{BQP} model to light by exploring its relationship to known quantum complexity classes and showing that it achieves nontrivial advantages over classical computation. Our results suggest that $\frac12$\textbf{BQP} is significantly more powerful than \textbf{DQC}1 and \textbf{IQP} despite suffering from some of the same drawbacks as these models. Most of our proofs come from applying standard techniques or extending known results from other models to the $\frac12$\textbf{BQP} model.

It was previously known due to \cite{aaronson2017computational} that  $\frac12$\textbf{BQP} can sample from the output distribution of \textbf{DQC}1 computations. In particular, $\frac12$\textbf{BQP} arises as a straightforward generalization of \textbf{DQC}1 in the following sense. Intuitively, the \textbf{DQC}1 model can be thought of as running a quantum circuit on a random input string, but with one clean qubit used to extract information, (or equivalently $O(\log(n))$ clean qubits). $\frac12$\textbf{BQP} is the same setup, but we learn the random string at the end of the computation. 
We can simulate a \textbf{DQC}1 circuit $C$ in $\frac12$\textbf{BQP} by running $C$ and optimistically assuming the first qubit is $\ket{0}$. This setup correctly simulates the \textbf{DQC}1 circuit with probability $1/2$. This probability can be made arbitrarily close to $1$ by running more copies of $C$ in parallel. 

\cite{aaronson2017computational} asked the following question: ``can we show that the power of $\frac{1}{2}$\textbf{BQP} is strictly intermediate between \textbf{DQC}1 and \textbf{BQP}.'' As we will explain throughout this section, in this work we give several pieces of evidence for why the answer is in the affirmative. We note that while a complete problem for \textbf{DQC}1 is estimating ``weighted'' normalized trace of a unitary matrix with weights that are computable by reversible circuits with few ancilla bits, $\frac{1}{2}$\textbf{BQP} can do the same where the weights are computable by arbitrary polynomial-time functions (\cref{thm:half-bqp-estimates-polytime-weighted-trace}). We also define a complete problem for $\frac{1}{2}$\textbf{BQP} which generalizes weighted trace (see \cref{thm:half-bqp-complete-problem}). We furthermore show that $\frac{1}{2}$\textbf{BQP} can perform many other tasks that are not known to be solvable by \textbf{DQC}1, such as \textsc{Forrelation} and \textsc{Factoring}. We also show an oracle separation with \textbf{BQP} and introduce other plausible candidates for such a separation. We will explain each of these results in the following.
 
We first show that $\frac12$\textbf{BQP} can sample from yet another well known intermediate model of computation, namely \textbf{IQP}.

\begin{theorem*} [Restatement of \cref{thm:half-bqp-simulates-iqp}]
    $\frac12$\textbf{BQP} can sample from \textbf{IQP}.
\end{theorem*}
\noindent To simulate \textbf{IQP}, we can think of a $\frac12$\textbf{BQP} circuit as a \textbf{BQP} circuit with random Pauli $X$ errors as the first layer of gates. Since \textbf{IQP} circuits are diagonal in the Pauli $X$ basis, this is equivalent to a \textbf{BQP} circuit with random $X$ errors in the last layer of gates. Thus, after learning the locations of these errors, we can correct them in post-processing. We remark that $\frac12$\textbf{BQP} can also simulate variations of \textbf{IQP} which are diagonal in any basis, as long as the change of basis is known and efficiently computable by a quantum circuit.
 
It is intriguing that $\frac12$\textbf{BQP} captures the power of both \textbf{DQC}1 and \textbf{IQP} despite these models arising due to very different resource restrictions. As a result of this, a separation between \textbf{DQC}1 and \textbf{IQP} would imply a separation between $\frac12$\textbf{BQP} and one of these models.
 
Next we study the relationship between $\frac{1}{2}$\textbf{BQP} and other computational models in the oracle setting.
We know that $\frac{1}{2}$\textbf{BQP} can sample from both \textbf{IQP} and \textbf{DQC}1 circuits. The following two theorems show that $\frac{1}{2}$\textbf{BQP} can solve several oracle problem problems which are not known to be solvable by either of the two models.

First, $\frac12$\textbf{BQP} can perform Fourier sampling, a task which is known to require exponentially many queries to perform on a class probabilistic machine.
\begin{theorem*} [Restatement of \cref{thm:fourier-sampling}]
    $\frac12$\textbf{BQP} can sample from the Fourier weight distribution of a boolean function. Concretely, given phase oracle access to a function 
    \[f:\{0,1\}^n \rightarrow \{-1,1\},\]
    $\frac12$\textbf{BQP} can sample from the distribution over $\{0,1\}^n$ such that
    \[\Pr[y] = \wh{f}(y)^2,\]
    where $\wh{f}(y)$ is the Fourier coefficient of $f$ corresponding to the string $y$ (See \cref{subsec:prelim-fourier-analysis} for a definition).
\end{theorem*}

Next, we show that $\frac12$\textbf{BQP} solves many of the standard oracle problems used to separate classical and quantum models of computation.

\begin{theorem*}[Restatement of \cref{thm:simons-problem}, \cref{thm:dj-bv} and \cref{thm:period-finding}]
    $\frac12$\textbf{BQP} solves the following oracle problems:
    \begin{itemize}
        \item Deutsch-Jozsa Problem
        \item Bernstein-Vazirani Problem
        \item Simon's Problem
        \item \textsc{Period Finding}
    \end{itemize}
\end{theorem*}
Formal proofs are given in \cref{sec:bqp-vs-bpp}. Intuitively, the proof of the above theorem is similar to the $\frac12$\textbf{BQP} algorithm to simulate \textbf{IQP}. It turns out that when translated to the $\frac12$\textbf{BQP} model, the standard algorithms for the above problems are equivalent to the standard \textbf{BQP} circuits with Pauli $X$ errors applied in the last layer of gates, which can be corrected in post-processing.

 
Having shown that $\frac12$\textbf{BQP} solves period finding, the next question to ask is whether $\frac12$\textbf{BQP} can solve \textsc{Order Finding} and \textsc{Factoring} outside of the oracle setting. While it is not immediately obvious how to implement the modular exponentiation step of Shor's order finding algorithm reversibly using few clean ancillas, it turns out that due to several nontrivial results from low-depth circuit complexity, the answer is affirmative.
\begin{theorem*}[Restatement of \cref{thm:order-finding}]
    \textsc{Order Finding} and \textsc{Factoring} are in  $\frac12$\textbf{BQP}.
\end{theorem*}

The next result is an oracle separation between $\frac12$\textbf{BQP} and the polynomial hierarchy (\textbf{PH}).
\cite{aaronson2014forrelation} showed that \textsc{Forrelation} problem gives an optimal separation between quantum and classical computations in the query complexity model. 
Subsequently, \cite{raz-tal} defined an explicit distribution which translates this separation to the oracle model, proving that there exists an oracle $O$ relative to which \textbf{BQP}$^O \not\subseteq$ \textbf{PH}$^O$. 
We show that, perhaps surprisingly, the quantum algorithms for \textsc{Forrelation} and the Raz-Tal problem can be performed in the $\frac12$\textbf{BQP} model with only a slight modification. This demonstrates that the largest oracle separation possible between quantum and classical computation holds even within this restricted model of quantum computation. This is captured in the following theorem:

\begin{theorem*} [Restatement of \cref{thm:2-forrelation} and \cref{thm:raz-tal}]
    \blankline
    \begin{enumerate}
        \item \textsc{Forrelation} is contained in $\frac12$\textbf{BQP}.
        \item There exists an oracle $O$ such that $\frac12$\textbf{BQP}$^O \not\subseteq$ \textbf{PH}$^O$.
    \end{enumerate}
\end{theorem*}

As a generalization of \textsc{Forrelation}, \cite{aaronson2014forrelation} defined $k$-\textsc{Forrelation} and showed that for $k = poly(n)$, $k$-\textsc{Forrelation} is complete for \textbf{BQP}. Since $2$-\textsc{Forrelation} is in $\frac12$\textbf{BQP}, then if $\frac12$\textbf{BQP} $\neq$ \textbf{BQP}, there exists some function $t(n)$ such that $t(n)$-\textsc{Forrelation} is not in $\frac12$\textbf{BQP}. This suggests a ``Forrelation hierarchy'' which captures the boundary between $\frac12$\textbf{BQP} and \textbf{BQP} and may be useful in finding an upper bound for $\frac12$\textbf{BQP}. We believe that this ``Forrelation hierarchy'' provides a meaningful chain of inclusions beginning with the power of Fourier Sampling and ending with \textbf{BQP}, and warrants further study for its own sake. We conjecture the boundary is at $k =3$, i.e., $\frac12$\textbf{BQP} cannot solve $3$-\textsc{Forrelation}.

The next result provides an oracle separation between $\frac12$\textbf{BQP} and \textbf{BQP} based on distinguishing unitary oracles which are close in Schatten-1 distance. This result shows that $\frac12$\textbf{BQP} suffers from the same drawbacks as \textbf{DQC}1 when distinguishing unitaries which are close in normalized trace. This result allows us to give an oracle separation using a classical boolean function based on a result from \cite{chen2023complexity}. Another consequence of this theorem is that  $\frac12$\textbf{BQP} cannot obtain the quadratic speedup for unstructured search given by Grover's algorithm (\cref{thm:grover}).

\begin{theorem*} [Informal statement of \cref{thm:bqp-separation}]
$\frac12$\textbf{BQP} cannot distinguish between unitary oracles $U$ and $U'$ such that $\frac{1}{2^n}\norm{U-U'}_1 \leq O(\frac{1}{2^n})$. Furthermore there exists a classical oracle (based on \cite{chen2023complexity}) that separates $\frac{1}{2}$\textbf{BQP} and \textbf{BQP}.
\end{theorem*}

We note that the oracle problem defined in \cite{chen2023complexity} also establishes an oracle separation between \textbf{DQC}1 and \textbf{BQP}. The proof of the theorem is a straightforward hybrid argument where we show that any $\frac12$\textbf{BQP} circuit distinguishing the two oracles requires a large number of queries. 
 
We also study what restrictions one can place on $\frac12$\textbf{BQP} so that it can be simulated by \textbf{DQC}1. The goal is to understand which resources make $\frac12$\textbf{BQP} stronger than \textbf{DQC}1. On the other hand, if we can find a resource that when added to \textbf{DQC}1 recovers $\frac12$\textbf{BQP}, this would place an upper bound on the power of $\frac12$\textbf{BQP}. Concretely, we show that \textbf{DQC}1 can solve promise-$\frac12$\textbf{BQP} complete problems when the $\frac12$\textbf{BQP} post-processing is sufficiently sparse in the Fourier basis.
\begin{theorem*}[Informal statement of \cref{thm:when-dqc1-simulates-half-bqp}]
    \textbf{DQC}1 can simulate a $\frac12$\textbf{BQP} computation when the $\frac12$\textbf{BQP} post-processing has a learnable Fourier spectrum.
\end{theorem*}

Overall, our results suggest that $\frac12$\textbf{BQP} is nontrivially positioned in the sub-universal quantum complexity landscape somewhere above certain known classes such as \textbf{DQC}1 and below \textbf{BQP}. They also suggest specific directions for further study with the goal of obtaining a clearer picture of the boundary between \textbf{BPP} and \textbf{BQP}.

\subsection{Related work}
\label{subsec:related-work}
We now briefly discuss some related works. Most similar to $\frac12$\textbf{BQP} is the \textbf{NISQ} model of computation \cite{chen2023complexity}, which captures noisy \textbf{BQP} circuits without fault tolerant error correction. It turns out that the \textbf{NISQ} complexity class yields very similar oracle separations to $\frac12$\textbf{BQP}. As an example, \cite{chen2023complexity} shows that \textbf{NISQ} solves a variation of Simon's problem in which the oracle is made robust against noise, giving a separation from \textbf{BPP}. More recently, \cite{chia2024oracle} showed that the Raz Tal oracle separation between \textbf{BQP} and \textbf{PH,raz-tal} translates to the \textbf{NISQ} model in the presence of depolarizing noise below a certain threshold, and when there is no noise after the first oracle query. This version of \textbf{NISQ} is similar to $\frac12$\textbf{BQP}, since they both study quantum circuits with noise in the beginning. In this sense, $\frac12$\textbf{BQP} and this version of \textbf{NISQ} seem to be close in computational power, and we are able to prove very similar results in both models.
 However, $\frac12$\textbf{BQP} and the version of \textbf{NISQ} studied in \cite{chia2024oracle} represent different scenarios and the proofs for the given oracle separations use different techniques. In particular, the \textbf{NISQ} model studies depolarizing noise below a certain threshold, where $\frac12$\textbf{BQP} studies gate based noise which occurs with probability $\frac{1}{2}$ (way above the noise threshold). Furthermore, $\frac12$\textbf{BQP} represents the scenario in which we learn the specific noise that occurred. It is also not clear how to solve \textsc{Period Finding} in \textbf{NISQ}. In particular, while $\frac{1}{2}$\textbf{BQP} can tolerate arbitrary high depth circuits for certain computations, \textbf{NISQ} cannot. On the other hand since \textbf{NISQ} circuits begin close to the all zeros state they have an advantage for low depth circuits (without any restriction on the commutation structure of the gates). We note that despite similarities to \cite{chia2024oracle}, our results were obtained independently.


\subsection{Open problems and future directions}
\label{subsec:open-problems}
We now give some open problems.
\begin{itemize}
    \item Is there an oracle $O$ relative to which \textbf{DQC}1$^O \subsetneq$ $\frac12$\textbf{BQP}$^O$?
    We believe that \textbf{DQC}1 is strictly contained in $\frac12$\textbf{BQP}, since we do not believe \textbf{DQC}1 can perform Fourier sampling or any of the related oracle problems. However, making this separation concrete is an open problem.
    
    \item What minimal resource can we provide to \textbf{DQC}1 so that it becomes equivalent to $\frac12$\textbf{BQP}?
    In \cref{subsec:when-dqc1-simulates-half-bqp} we will see that we can represent a $\frac12$\textbf{BQP} computation as a weighted sum of \textbf{DQC}1 computations. Finding an oracle which strengthens \textbf{DQC}1 to $\frac12$\textbf{BQP} would place an upper bound on the power of $\frac12$\textbf{BQP} and help clarify the exact relationship between \textbf{DQC}1 and $\frac12$\textbf{BQP}.

    \item What is the relationship between $\frac12$\textbf{BQP} and other intermediate quantum models of computation?
    Since we have shown that $\frac12$\textbf{BQP} contains certain intermediate quantum models, it is reasonable to ask about other models such as \textsc{BosonSampling} or \textbf{NISQ}. Can we show containments, or that these models are incomparable? We conjecture that $\frac12$\textbf{BQP} is incomparable to \textbf{NISQ} and \textbf{BPP}$^{\textbf{QNC}1}$.

    \item Consider a generalization of \textbf{NISQ} where we get to learn something about the content of noise (e.g. by measuring the environment surrounding the device) at the end of the computation. Can we show that by adding this modification we obtain a model that is strictly more powerful than \textbf{NISQ}? $\frac{1}{2}$\textbf{BQP} can be viewed as an instance where an extremely high rate of noise is applied in the beginning. 
    
    \item Can $\frac12$\textbf{BQP} solve $3$-\textsc{Forrelation}?
    We conjecure that $\frac12$\textbf{BQP} cannot solve $3$-\textsc{Forrelation} and give some intuition for why this might be the case.
    \item More broadly, is there some $t(n)$ such that $\frac12$\textbf{BQP} cannot solve $t(n)$-\textsc{Forrelation} for some $t(n)$?
    If our conjecture is false, it is still important to know where $\frac12$\textbf{BQP} sits in the Forrelation hierarchy. Resolving this problem is important for understanding the limits of $\frac12$\textbf{BQP} and its relationship to \textbf{BQP}.
    \item Can $\frac12$\textbf{BQP} implement the quantum Fourier transform over arbitrary cyclic groups? While $\frac12$\textbf{BQP} can implement the quantum Fourier transform over cyclic groups over a modulus which is a power of $2$, we have not been able to implement the quantum Fourier transform for arbitrary cyclic groups. In \cref{subsubsec:hsp} we show that if this is possible, then $\frac12$\textbf{BQP} can solve the arbitrary hidden subgroup problem for abelian groups.
    \item Can $\frac12$\textbf{BQP} implement the quantum phase estimation procedure? If this is the case, then $\frac12$\textbf{BQP} can implement Kitaev's algorithm for the quantum Fourier transform over arbitrary cyclic groups.
\end{itemize}

\subsection*{Acknowledgements}

S. M.\ and D. J.\ are grateful to the National Science Foundation (NSF CCF-2013062) for supporting this project. We thank V. Podolskii, A. Blumer, A. Motamedi, J. Jeang, M. Joseph, W. Khangtragool, M. Coudron, L. Schaeffer, U. Chabaud, A. Bouland, G. Kuperberg, and S. Aaronson for insightful conversations. 

\section{Preliminaries}

\subsection{Complexity theory background}
\subsubsection*{Computational complexity}
In computational complexity theory, we are interested in characterizing computational models based on which problems they can solve under resource constraints, usually space or time. Formally, a \textit{complexity class} $A$ is the set of decision problems solvable by a given computational model, or equivalently the set of languages decidable by the model. Throughout this work, we will sometimes abuse terminology by stating that a non-decision problem is ``solved" by a decision class (such as $\frac12$\textbf{BQP} solving \textsc{Factoring}). These cases should be clear from context. Furthermore, unless stated otherwise, all of the problems discussed have equivalent statements as decision problems. 
 
Characterizing complexity classes generally takes one of two forms: containment, or separation.
Showing a containment $A \subseteq B$ usually involves describing a procedure that allows the model describing B to efficiently simulate the model describing A.
On the other hand, showing a separation $A \neq B$ involves finding a problem that is in either A or B but not both.
Characterizing complexity classes is practically relevant in that it determines what types of computational problems are practically feasible under a given model of computation.
 
A relevant subset of complexity theory is quantum complexity theory, 
which deals with models of computation that leverage quantum effects.
An important open question in this area is whether quantum models of computation are provably more powerful than classical models.
It is widely believed that quantum computers should provide a theoretical computational advantage,
and so there has been much research attempting to prove, or at least provide evidence for this advantage.
 
We now give informal descriptions of several standard complexity classes which are relevant to our discussion. For a formal treatement, we recommend the Complexity Zoo \cite{aaronson2005complexity}.
\subsubsection*{Standard complexity classes}
\textbf{P} is the set of decision problems solvable by a polynomial-time Turing machine. Similarly, \textbf{BPP} is the set of decision problems solvable by a polynomial-time Turing machine with access to random bits. In general, \textbf{P} and \textbf{BPP} are considered to represent problems which can be solved efficiently by practically feasible models of classical computing.
 
\textbf{NP} is the set of decision problems corresponding to existential quantifiers over polynomial-time computable boolean functions. That is, \textbf{NP} problems ask whether there exists a string $x \in \{0,1\}^n$ such that $f(x) = 1$, where $f$ is computable by a polynomial-time Turing machine.
 
Similarly, \textbf{Co-NP} is the set of decision problems corresponding to universal quantifiers over polynomial-time computable boolean functions. \textbf{Co-NP} problems ask whether for all $x \in \{0,1\}^n$, $f(x) = 1$, where again $f$ is computable by a polynomial-time Turing machine.
 
The Polynomial Hierarchy \textbf{PH} defines a generalization of \textbf{NP} and \textbf{Co-NP} corresponding to combinations of existential and universal quantifiers over polynomial-time computable boolean functions. In particular, \textbf{PH} is defined as
\[\textbf{PH} = \bigcup_i \bm{\Sigma_i}\]
where $\bm{\Sigma_i}$ corresponds to problems of the form
\[\exists x_1, \forall x_2, \dots, Q_i x_{i}, f(x_1, x_2,\dots,x_i) = 1.\] 
Here $Q_i$ is either $\exists$ or $\forall$ depending on whether $i$ is even or odd, and again $f$ is computable by a polynomial-time Turing machine. Importantly, \textbf{PH} is a particularly strong complexity class, widely believed to be significantly more powerful than any physically realizable model of computation. Furthermore, it is a widely believed conjecture that \textbf{PH} is not equivalent to any of its finite levels, that is \textbf{PH} $\neq \bm{\Sigma_i}$ for any finite $i$. 
 
\textbf{BQP} is the set of decision problem solvable by uniformly-generated polynomial-size quantum circuits with error probability $\leq \frac{1}{3}$. Informally, \textbf{BQP} represents problems that are efficiently solvable by a universal quantum computer. Importantly, \textbf{BQP} contains \textbf{BPP}.

\subsubsection*{Oracles and Oracle Separations}
An \textit{oracle} can be thought of as a black-box solving some decision problem or implementing some function in one time step,
and an \textit{oracle model} is a model with query access to some oracle.
Concretely, oracle access means that the model may query the oracle to obtain the correct output to any instance of the corresponding problem or function in one time step. For classical computations, this is usually straightforward. An oracle $O_f$ implementing a function $f$ outputs $f(x)$ when queried on input $x$. However, to use oracles in a quantum circuit, we must be careful about how the function is implemented as a unitary transformation. The two most common forms of quantum oracle are the phase oracle and the index oracle, which we now define.
\begin{definition}[Phase oracle]
    Let $f: \{0,1\}^n \rightarrow \{0,1\}$ and let $O_f$ be a phase oracle implementing $f$. Then $O_f$ acts according to
    \[O_f: \ket{x} \mapsto (-1)^{f(x)}\ket{x},\]
    where $x \in \{0,1\}^n$.
\end{definition}
\begin{definition}[Index oracle]
    Let $f: \{0,1\}^n \rightarrow \{0,1\}^m$ and let $O_f$ be an index oracle implementing $f$. Then $O_f$ acts according to
    \[O_f: \ket{x}\ket{y} \mapsto \ket{x}\ket{y \oplus f(x)},\]
    where $x \in \{0,1\}^n$ and $y \in \{0,1\}^m$.
\end{definition}
The class of problems solvable by the model describing class A with access to oracle O is denoted A$^\text{O}$.
If there exists a problem which is in A$^\text{O}$ but not B$^\text{O}$, 
we say that A and B are separated relative to O. 
This is called an \textit{oracle separation}.
Note that A$^\text{O} \neq$ B$^\text{O}$ does not necessarily imply A $\neq$ B.
Furthermore A $\neq$ B does not always imply A$^\text{O} \neq$ B$^\text{O}$, depending on the choice of $A$, $B$ and $O$.
 
There are several well-known oracle separation results relating \textbf{P} and \textbf{BQP}.
An early result is the Deutsch-Jozsa problem, followed by Simon's problem \cite{deutsch1992rapid,simon1997power}.
The Deutsch-Jozsa gives an oracle separation between \textbf{BQP} and \textbf{P}, and Simon's problem gives an oracle separation between \textbf{BQP} and \textbf{BPP}. 
More recently, Raz and Tal \cite{raz-tal} gave an oracle separation between \textbf{BQP} and \textbf{PH} based on \textsc{Forrelation} \cite{aaronson2014forrelation}.
 
In most cases, the oracle O is not efficiently implementable nor practically relevant.
However, oracle separations are still useful for studying relationships between complexity classes.

\subsubsection*{Asymptotics}
Following the standard convention, we use $O(\cdot), \Omega(\cdot)$ to hide universal constants.

\subsection{The one-clean-qubit model (\textbf{DQC}1)}

Inspired by nuclear magnetic resonance (NMR) quantum computation, the one-clean-qubit model was defined in \cite{Knill_1998} to capture the computational power of one bit of quantum information. In this model, we run a uniformly generated polynomial-size quantum circuit $C$ on $n+1$ registers, consisting of $1$ clean qubit and $n$ qubits in the maximally mixed state. Formally, the input to the circuit is described by
\[\rho = \ketbra{0}{0} \otimes \frac{I}{2^n}.\]
The \textbf{DQC}1 complexity class is the set of problems solvable by polynomially many \textbf{DQC}1 circuits with \textbf{BPP} pre and post-processing.
 
It is important that no intermediate measurements are allowed,
otherwise we could initialize the input state to all zeros and recover \textbf{BQP}.
 
\textbf{DQC}1 is well studied and several important results are known about its computational power. It is believed to be strictly intermediate between \textbf{BPP} and \textbf{BQP}, though this is still open. One important result is that \textit{Normalized Trace Estimation} is \textbf{DQC}1-complete.
\begin{fact}[\cite{Knill_1998}]
    \label{fact:dqc1-trace}
    Given a description of a quantum circuit $C$ implementing unitary $U$, estimating the quantity
    \[\frac{1}{2^n}\Tr(U)\]
    within inverse polynomial additive error is complete for  \textbf{DQC}1.
\end{fact}
Specifically, the above statement means that we can simulate \textbf{DQC}1 circuits by estimating the normalized trace of a unitary circuit, and there is a \textbf{DQC}1 computation that can estimate the normalized trace within additive error. Another useful fact is that \textbf{DQC}1 circuits are not strengthened by using up to $\log(n)$ clean qubits.
\begin{fact}[\cite{ambainis2000computing, shor2008estimating}]
    \label{fact:dqc1-dqclogn}
    \textbf{DQC}(k) = \textbf{DQC}1 for $k = O(\log(n))$.
\end{fact}
One more fact is that \textbf{BPP} cannot simulate \textbf{DQC}1 circuits unless the polynomial hierarchy collapses.
\begin{fact}[\cite{ambainis2000computing}]
    \label{fact:bpp-cant-sample-dqc1}
    \textbf{BPP} cannot sample from \textbf{DQC}1 circuits unless \textbf{PH} collapses to its third level.
\end{fact}

Intuitively, the \textbf{DQC}1 model can be thought of as running a quantum circuit on a random input string,
but with one clean qubit (equivalently $O(\log n)$ clean qubits) used to extract information.
This perspective works well to capture the \textbf{DQC}1-complete problem of normalized trace estimation,
in which we use our clean qubit in the Hadamard test to sample random elements from the diagonal of a unitary.
 
However, there is a significant drawback in the one-clean-qubit model compared to \textbf{BQP}.
Although we can run any circuit we want on our random input string,
there are many quantum algorithms which we cannot implement in \textbf{DQC}1,
since their implementation may require a superlogarithmic number of clean ancilla qubits. In fact, it is not even clear how to simulate classical computation using \textbf{DQC}1 circuits. \footnote{As a consequence, if we define a variation of \textbf{DQC}1 without polynomial-time classical pre-processing, this variation does not necessarily contain \textbf{P}.}.
Since many quantum algorithms use classical algorithms as subroutines, \textbf{DQC}1 is significantly disadvantaged by its limited number of clean ancillae compared to \textbf{BQP}.
 
Due to this limitation, it is meaningful to ask which classical functions are computable by \textbf{DQC}1 circuits. Clearly \textbf{DQC}1 can implement functions computable by reversible circuits with up to $O(\log n)$ clean ancillae. Notably, this class of functions contains logarithmic-depth circuits (\textbf{NC}$^1$). This is a consequence of Barrington's theorem \cite{barrington1986bounded, ambainis2000computing}. Thus, while \textbf{DQC}1 circuits may not use arbitrary classical functions as subroutines, they can use \textbf{NC}$^1$ functions as subroutines, or more generally any classical function computable by reversible circuits with a logarithmic number of clean ancillae, as previously stated.
 
Using classical functions as a subroutine in \textbf{DQC}1 circuits leads to a generalization of normalized trace estimation which has proven useful for designing \textbf{DQC}1 algorithms and will be relevant to our discussion. 

\subsubsection*{Weighted Normalized Trace Estimation}
So far we have seen that normalized trace estimation is a DQC1-complete problem.
We have also seen that DQC1 = DQC($\log n$).
Combining these two facts, 
we see that DQC1 can compute a weighted trace over a subset of $n$-bit strings.
\begin{fact}[\cite{shor2008estimating}]
  Let $f: \{0,1\}^n \rightarrow [-1,1]$ be a function with a polynomial-size reversible circuit implementation using $O(\log n)$ ancillas.
  Then it is \textbf{DQC}1-complete to estimate the f-weighted normalized trace of a circuit C within inverse polynomial additive error, 
  where the f-weighted normalized trace is defined as
  \[W_f\Tr(C) = \frac{1}{2^n}\sum_{x \in \{0,1\}^n} f(x) \bra{x}C\ket{x}\]
\end{fact}
Note that for $f$-weighted trace estimation to be meaningful, 
$f$ must be nonzero on at least a polynomial fraction of all input strings.
Otherwise, the weighted trace will be exponentially close to $0$.
We can think of the quantity as a weighted trace over a sufficiently large subspace of our Hilbert space.
We now give a proof to illustrate the \textbf{DQC}1 procedure for estimating $W_f \Tr(C)$.
\begin{proof}
  Consider the following circuit which first computes $f$ and then implements the Hadamard test. Here the $f$ gate implements
  \[ f: |0^{O(\log n)}\rangle\ket{j} \mapsto \ket{f(j)}\ket{j}. \]
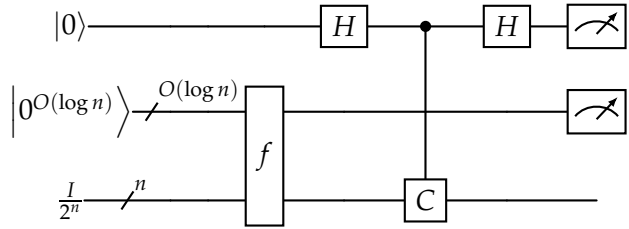
\begin{figure}[H]
  \centering
  \begin{quantikz}
      \ket{0} &&&& \gate{H} & \ctrl{2} & \gate{H} & \meter{} \\
      \ket{0^{O(\log n)}} & \qwbundle{O(\log n)} && \gate[2]{f} &&&& \meter{} \\
      \frac{I}{2^n} & \qwbundle{n} &&&& \gate{C} &&
  \end{quantikz}
\end{figure}
For each run of the circuit, 
the first register outputs $\ket{0}$ with probability $\frac{1}{2} + \frac{\Re(\bra{x}C\ket{x})}{2}$,
where $x$ is a random input string,
and the next $O(\log n)$ registers output $f(x)$ deterministically.
Thus, we can estimate the real component of the $f$-weighted trace.
Similarly, using the standard procedure for the Hadamard test, we can compute the imaginary component.
\end{proof}
Note that since the $f$-weighted normalized trace is a generalization of the normalized trace,
this problem is DQC1-complete.
We will see in \cref{thm:half-bqp-estimates-polytime-weighted-trace} that $\frac12$\textbf{BQP} can compute a stronger version of the weighted normalized trace where $f$ is any polynomial-time classical function. 

\subsection{Instantaneous quantum circuits (\textbf{IQP})}

Another relevant sub-universal model is Instantaneous Quantum Polynomial-time (\textbf{IQP}), defined in \cite{Shepherd_2009,shepherd2010quantum}. \textbf{IQP} is a restricted version of the quantum circuit model and has been studied as a candidate for quantum advantage experiments.
 
An \textbf{IQP} circuit is a quantum circuit $C$ consisting of gates which are diagonal in the Pauli $X$ basis.
Equivalently, via a change of basis, we can think of an \textbf{IQP} circuit as $C = H^{\otimes n}DH^{\otimes n}$, where $D$ is diagonal in the Pauli $Z$ basis. For our purposes, we will adopt this definition. The complexity class \textbf{IQP} as the set of problems solvable by an \textbf{IQP} circuit with polynomial-time pre and post-processing.
\begin{figure}[H]
    \centering
    \begin{quantikz}
        \ket{0} & \gate{H^{\otimes n}} & \gate{D} & \gate{H^{\otimes n}} & \meter{}
    \end{quantikz}
    \caption{An \textbf{IQP} circuit. $D$ is diagonal.}
\end{figure}
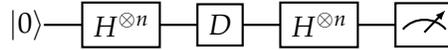
\textbf{IQP} is believed to be intermediate between \textbf{BPP} and \textbf{BQP}, and similarly to \textbf{DQC}1, \textbf{BPP} cannot simulate \textbf{IQP} unless the polynomial hierarchy collapses.
\begin{fact} [\cite{bremner2011classical}]
    \textbf{BPP} cannot sample from \textbf{IQP} circuits unless \textbf{PH} collapses to its third level.
\end{fact}

\subsection{Fourier analysis of boolean functions}
\label{subsec:prelim-fourier-analysis}
\subsubsection*{Fourier representation basics}
Let $f: \{-1,1\}^n \rightarrow \mathbb{R}$. The Fourier representation of $f$ is given by
\begin{align*}
    f(x) = \sum_{S \subset [n]} \wh{f}_S \cdot \chi_S(x).
\end{align*}
where
\begin{align*}
    \chi_S = \prod_{i \in S} x_i.
\end{align*}
and 
\begin{align*}
    \wh{f}_S = \langle f, \chi_S \rangle = \frac{1}{2^n} \sum_{x} f(x) \chi_S.
\end{align*}
\noindent We call $\wh{f}_S$ the Fourier coefficient of $f$ corresponding to the set $S$.

For our purposes it will be convenient to represent $f$ as a function on the boolean cube, $f: \{0,1\}^n \rightarrow \mathbb{R}$, and we will use this notation going forward. In this case, 
\begin{align*}
    \chi_s(x) = (-1)^{x \cdot s}
\end{align*}
where $s$ is the characteristic string representing some set $S \subset [n]$ That is, $s_i = 1$ if $i \in S$ and $s_i = 0$ otherwise. Here $\cdot$ denotes the dot product over $\mathbb{F}_2^n$, i.e., $x \cdot y = x_1 \cdot y_1 + \ldots + x_n \cdot y_n \mod 2$, for $x, y \in \mathbb{F}_2^n$.

In this notation,
\begin{align*}
    f(x) = \sum_{s \in \{0,1\}^n} \wh{f}_s (-1)^{x \cdot s}
\end{align*}
where 
\[\wh{f}_s = \frac{1}{2^n} \sum_{x \in \{0,1\}^n} f(x)(-1)^{x \cdot s}.\]
We will sometimes use $\wh{f}(s)$ to represent $\wh{f}_s$.

\subsubsection*{Learning Fourier coefficients}

We now give some important facts about learning Fourier coefficients, which will be useful for our analysis of \cref{thm:when-dqc1-simulates-half-bqp}.

\begin{fact}[See \cite{o2014analysis}]
    \label{fact:learning-fourier}
    Given query access to $f: \{0,1\}^n \rightarrow [-1,1]$, for any $s \in \{0,1\}^n$ we can estimate $\wh{f}_s$ within $\pm \varepsilon$ with probability $1 - \delta$ in time $poly(n,\frac{1}{\varepsilon^2})\log(\frac{1}{\delta})$.
\end{fact}

\begin{fact}[Goldreich-Levin Theorem \cite{goldreich1989hard}]
    \label{fact:goldreich-levin}
    Given query access to $f: \{0,1\}^n \rightarrow [-1,1]$ and a parameter $0 < \tau \leq 1$, there is a $poly(n, \frac{1}{\tau})$ algorithm that with high probability outputs a list $L = \{s_1, s_2,\dots, s_l\}, |L| \leq \frac{4}{\tau^2}$ containing all $s_i$ such that 
    $|\wh{f}_s| \geq \tau$. (Here $f(x)$ is rational and expressible using $poly(n)$ bits).
\end{fact}

\section{Four Definitions for $\frac{1}{2}$\textbf{BQP} }
\label{sec:three-definitions}

\subsection{$\frac{1}{2}$\textbf{BQP} as a model of computation on entangled states }

In \cite{aaronson2017computational}, $\frac12$\textbf{BQP} was defined to be the following model of quantum computing: We start with $2n$ qubits initially in the maximally entangled state:
\[\frac{1}{\sqrt{2^n}}\sum_{x \in \{0,1\}^{n}}\ket{x}\ket{x}.\]
We then apply a polynomial-size quantum circuit from a \textbf{BPP} uniform family of quantum circuits to the left-most $n$ qubits but not to the other $n$ qubits.
At the end of the computation,
we measure all $2n$ qubits and perform \textbf{BPP} post-processing. No intermediate measurements are allowed (or else we trivially recover \textbf{BQP}).
We call the first $n$ qubits the \textit{left} qubits and the last $n$ qubits the \textit{right} qubits.
Note that if we could only measure the left half of the state but not the right half, the model would be equivalent to computation over the maximally mixed state and would not be useful for performing computation. This means that the \textbf{BPP} post-processing plays a central role in the power of $\frac12$\textbf{BQP}. This is in contrast to \textbf{BQP} computations, where \textbf{BQP} algorithms can be used as subroutines within \textbf{BQP} circuits (\textbf{BQP}$^\textbf{BQP}$ = \textbf{BQP}). As a result of this observation, it is unlikely that $\frac12$\textbf{BQP}$^{\frac12\textbf{BQP}}$ = $\frac12$\textbf{BQP}.
 
As another important point, because we do not know $w$ until after running the $\frac12$\textbf{BQP} circuit, we are only afforded a small number of clean qubits in our computation. This is prohibitive for the $\frac12$\textbf{BQP} model because there are many \textbf{BQP} algorithms that use pure ancilla qubits to implement classical functions reversibly. Concretely, optimistically assuming the first $k$ qubits are $\ket{0}$ and verifying at the end yields $O(\log(n))$ pure ancilla qubits in the $\frac12$\textbf{BQP} model with probability $\frac{1}{poly(n)}$. It is not clear how to obtain more pure ancilla qubits.

$\frac{1}{2}$\textbf{BQP} is the following family of decision problems. 
\begin{definition}[Definition based on computation on entangled states]
    The $\frac12$\textbf{BQP} complexity class is the set of languages $L \subseteq \{0,1\}^*$ for which there exists a \textbf{BPP} uniform family of $\frac{1}{2}$\textbf{BQP} circuits based on computation on one-half of maximally entangled states as outlined above such that for any $x \in \{0,1\}^*$, if $x \in L$, the computation accepts with probability $\geq \frac{2}{3}$ and otherwise rejects with probability $\geq \frac23$.
\end{definition}
Going forward, we use $\frac{1}{2}$\textbf{BQP} to both mean the model of computation and the set of decision problems solvable by this model. 
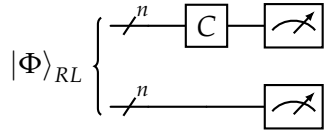
\begin{figure}[H]
        \centering
        \begin{quantikz} 
            \lstick[2]{$\ket{\Phi}_{RL}$} & \qwbundle{n} & \gate{C}  & \meter{} \\
            \qw & \qwbundle{n} && \meter{}
        \end{quantikz}
        \caption{A $\frac12$\textbf{BQP} circuit, where $\ket{\Phi}_{RL} = \frac{1}{\sqrt{2^n}} \sum_{x \in \{0,1\}^n} \ket{x}\ket{x}$}
\end{figure}

\subsection{$\frac12$\textbf{BQP} as computation over random basis states}
\label{subsec:random-bqp}
We now give an equivalent formulation of $\frac12$\textbf{BQP}, which is framed in terms of quantum computation over random basis states.
\begin{definition}[Quantum Computation over Random Strings (\textbf{Random-BQP})]
    Suppose we have an $n$-qubit quantum register initialized to a uniformly random computational basis state, $\ket{w}$ for $w \in \{0,1\}^{n}$. We don't know $w$ ahead of time.
We are allowed to run a uniformly generated quantum circuit $C$ without intermediate measurements. At the end of the circuit, we sample from $C\ket{w}$ and we also learn $w$. \textbf{Random-BQP} is the class of problems solvable by the above model with polynomial time classical pre and post-processing.
\end{definition}
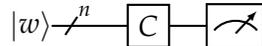
\begin{figure}[H]
    \centering
    \begin{quantikz}
        \ket{w} & \qwbundle{n} & \gate{C} & \meter{}
    \end{quantikz}
    \caption{An equivalent circuit for $\frac12$\textbf{BQP}. At the end of the computation we learn $w$.}
\end{figure}
\begin{claim}
    \textbf{Random-BQP} = $\frac12$\textbf{BQP}.
\end{claim}
\begin{proof}
    To see equivalence of the two models, consider the case where we measure the right half of the $\frac12$\textbf{BQP} state at the beginning of the computation but cannot modify the circuit $C$ based on this measurement outcome. The measurement collapses the left half into a basis state $\ket{w}$ according to a uniform distribution. We then run circuit $C$ on the left half and at the end obtain a sample from $C\ket{w}$, as well as the previously measured $w$.
    Since in the $\frac12$\textbf{BQP} circuit we do not apply any gates to the right half, this is the same as measuring both halves at the end.
\end{proof}
This formulation is often more convenient for proving results and in general we will prefer it over the bipartite formulation.
 
There is one more formulation of $\frac12$\textbf{BQP} that will be useful. Consider a standard \textbf{BQP} circuit with input $\ket{0^n}$, except the first layer of the circuit consists of the gate $X^w$ for some uniformly random $w \in \{0,1\}^n$, where
\[X^w = X^{w_1} \otimes X^{w_2} \otimes \dots \otimes X^{w_n}.\]
Again we learn the string $w$ at the end of the computation. It is easy to see that this is equivalent to a $\frac12$\textbf{BQP} computation over a random string. This formulation is useful because it sometimes allows us to take advantage of commutation relations by ``passing'' the $X$ gates forward through the circuit (see \cref{subsec:half-bqp-simulates-iqp} for an example).
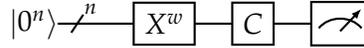
\begin{figure}[H]
    \centering
        \begin{quantikz}
            \ket{0^n} & \qwbundle{n} & \gate{X^w} & \gate{C} & \meter{}
        \end{quantikz}
    \caption{Another equivalent $\frac12$\textbf{BQP} circuit. At the end of the computation we learn $w$.}
\end{figure}


\subsection{$\frac12$\textbf{BQP} as a generalization of one-clean-qubit}

$\frac12$\textbf{BQP} can be thought of as a generalization of the \textbf{DQC}1 (one-clean-qubit) model defined in \cite{Knill_1998}. 
Intuitively, the \textbf{DQC}1 model can be thought of as running a quantum circuit on a random input string, but with one clean qubit used to extract information, or equivalently $\log(n)$ clean qubits. This is in contrast to the $\frac12$\textbf{BQP} model where we learn the entire random input string at the end of the circuit. It is not hard to see that $\frac12$\textbf{BQP} can simulate \textbf{DQC}1, and this was shown in \cite{aaronson2017computational}. For completeness, we give the proof here.
\begin{fact}
    \label{fact:halfbqp-simulates-dqc1}
    $\frac12$\textbf{BQP} can simulate \textbf{DQC}1 \cite{aaronson2017computational}.
\end{fact}
\begin{proof}
    To simulate a \textbf{DQC}1 computation in $\frac12$\textbf{BQP} we optimistically assume that the first bit of $w$ is $0$, then run the \textbf{DQC}1 circuit. Since $w$ is uniformly random, the probability of successfully simulating the \textbf{DQC}1 computation in one attempt is $1/2$, which can be amplified by repeating the procedure in parallel. Concretely, if we run $k$ copies of the \textbf{DQC}1 circuit in parallel, the probability that none of them correctly simulate the \textbf{DQC}1 circuit is $\frac{1}{2^k}$.
\end{proof}

As an immediate corollary to \cref{fact:halfbqp-simulates-dqc1} and \cref{fact:bpp-cant-sample-dqc1}, \textbf{BPP} cannot sample from $\frac12$\textbf{BQP} circuits unless \textbf{PH} collapses. As an alternative proof, if we postselect on $w$ being the all zeros state, as well as postselect on the output qubit of our computation, we recover \textbf{PostBQP}. It is well known due to \cite{aaronson2005quantum} that this implies \textbf{BPP} cannot sample from $\frac12$\textbf{BQP} circuits unless \textbf{PH} collapses.

We have seen that $\frac12$\textbf{BQP} can simulate arbitrary \textbf{DQC}1 computations. Can we show that $\frac12$\textbf{BQP} is strictly stronger than \textbf{DQC}1? Recall that a complete problem for \textbf{DQC}1 is giving additive estimate to a weighted trace where the weight function has the restriction that it has a reversible classical circuit that does not use more than logarithmic number of ancilla bits. In this Section we first show that the $\frac12$\textbf{BQP} model can solve the weghted trace problem without this restriction. Then in section \cref{subsec:half-bqp-complete-problem} we give a complete problem for $\frac12$\textbf{BQP} which includes the general weighted trace problem as a special case. While this result does not prove a separation between $\frac12$\textbf{BQP} and \textbf{DQC}1, it is the first of many compelling pieces of evidence that $\frac12$\textbf{BQP} is a stronger model of computation than \textbf{DQC}1.
\begin{theorem}
  \label{thm:half-bqp-estimates-polytime-weighted-trace}
  Let $f: \{0,1\}^n \rightarrow [-1,1]$ be computable by any polynomial-time (\textbf{BPP}) classical algorithm.
  Then $\frac12$\textbf{BQP} can estimate the f-weighted normalized trace of a circuit C within inverse polynomial additive error.
\end{theorem}
\begin{proof}
    Recall that the $f$-weighted normalized trace is given by 
    \[W_f\Tr(C) = \frac{1}{2^n}\sum_{x \in \{0,1\}^n} f(x) \bra{x}C\ket{x}.\]
    Consider the following $\frac12$\textbf{BQP} circuit which implements the Hadamard test.
    
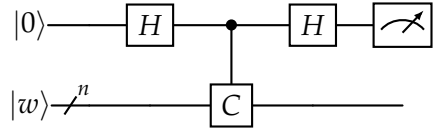
\begin{figure}[H]
        \centering
        \begin{quantikz}
            \ket{0} && \gate{H} & \ctrl{1} & \gate{H} & \meter{} \\
            \ket{w} & \qwbundle{n} && \gate{C} && 
        \end{quantikz}
    \end{figure}
    The probability of measuring $0$ on the first register for a fixed $w$ is 
    \begin{align}
        \Pr(0) = \frac{1}{2} + \frac{\Re(\bra{w}U\ket{w})}{2}.
    \end{align}
    Since we learn $w$, we can compute $f(w)$ in post-processing.
    Let $A$ be the expectation of $f$ over all runs of the circuit which output $0$ on the first register. Then
    \begin{align}
    \begin{split}
        A &= \frac{1}{2^n} \sum_{w \in \{0,1\}^n} f(w) \Bigl( \frac{1}{2} + \frac{\Re(\bra{w}U\ket{w})}{2} \Bigr) \\
        &= \frac{1}{2^{n+1}} \sum_{w \in \{0,1\}^n} f(w) + \frac{1}{2^n} \sum_{w \in \{0,1\}^n} f(w) \frac{\Re(\bra{w}U\ket{w})}{2} \\
        &= \frac{1}{2}\underset{w \in \{0,1\}^n}{\E}[f(w)] + \frac{1}{2^{n+1}} \sum_{w \in \{0,1\}^n} f(w) \Re(\bra{w}U\ket{w})
    \end{split}
    \end{align}
    By a standard Chernoff bound, we can estimate $2A$ and $\underset{w \in \{0,1\}^n}{\E}[f(w)]$ within inverse polynomial additive error. Taking their difference yields an additive approximation to the real component of $W_f\Tr(C)$. To estimate the imaginary component, we use the same procedure with the standard modification to the Hadamard test circuit.
\end{proof}
Thus, where the \textbf{DQC}1 weighted trace function must be efficiently computable by a reversible circuit with $O(\log n)$ ancilla bits, $\frac12$\textbf{BQP} can compute a weighted trace using any polynomial-time computable function.

\subsection{$\frac12$\textbf{BQP} as a ball permutational model}
\label{subsec:ball-perm}

The third definition of $\frac12$\textbf{BQP} is based on the power of quantum computations based on gates that allow many conserved quantities. Consider a gate set from an algebra $\mathcal{A}$ acting on a Hilbert space $\mathcal{H}$ that has a nontrivial commutant $\mathcal{E}$ (set of operators that commute with $\mathcal{A}$; nontrivial means it contains elements that are not proportional to $I$). We can think of $\mathcal{E}$ as noise, and we want to use gates from $\mathcal{A}$ to perform computations that are protected against operators in $\mathcal{E}$. Previous work \cite{lidar2014review} indicates that one can encode universal quantum computations in special ``decoherence-free'' subspaces (corresponding to irreducible subspaces of $\mathcal{A}$) that are immune from noise. In particular, under the action of $\mathcal{A}$ and $\mathcal{E}$ we can decompose the Hilbert space $\mathcal{H} \cong \bigoplus_{\lambda} V_\lambda \otimes W_\lambda$, where elements of $\mathcal{A}$ only affect the left tensor product $V_\lambda$ and every element of $\mathcal{E}$ affets the right side $W_\lambda$. Each $V_\lambda$ can be viewed as a decoherence-free subspace. If we prove that the action of the gate set generates a dense image inside $V_\lambda$ with large enough dimensionality, assuming details about programmability can be worked out, we can perform universal quantum computation if we start from and measure in the $V_\lambda$ basis. $\frac12$\textbf{BQP} models the situation where we start the computation from input states that are entangled across decoherence-free subspaces ($V_\lambda$) and subspaces that are subject to decoherence ($W_\lambda$), and we get to measure the content of both subspaces in the end. 

One concrete example of this abstract scenario was worked out in \cite{aaronson2017computational} in the context of studying the so-called quantum ball permutations. In the quantum ball permutational model, the Hilbert space is spanned by an orthonormal basis of $n!$ permutations of an $n$ element set (e.g., permutations of $n$ colored balls placed on a line). We denote this Hilbert space by $\mathbb{C} [S_n]$, where $S_n$ is the symmetric group. In this model, we start with the arrangement $|12\ldots n\rangle$ of particles. At each step of the computation, we get to pick a pair of particles and apply a permutational gate $\ket{ab} \mapsto c \ket{ab} + is \ket{ba}$. Each of these gates belongs to the algebra of left permutations $\mathcal{A}$, meaning applying a permutational element from left by permuting the physical location of particles. In the language of representation theory, this corresponds to the regular left representation of the symmetric group $L : S_n \rightarrow End (\mathbb{C} [S_n])$. For a permutation basis $\ket{\sigma} \in \mathbb{C} [S_n]$ and permutation $\tau \in S_n$, $L_\sigma : \ket{\sigma} \mapsto \ket{\tau \circ \sigma}$ (where $\tau \circ \sigma$ is the composition of $\tau$ with $\sigma$ from the left).  $\mathcal{E}$ corresponds to the set of right permutations, meaning permuting the labels of the particles (e.g., identifying $1$ with $2$ and $2$ with $1$ elsewhere). In other words, each element of $\mathcal{E}$ is a linear combination (with complex numbers) of right representations $R : S_n \rightarrow End (\mathbb{C} [S_n])$. For a permutation basis $\ket{\sigma} \in \mathbb{C} [S_n]$ and permutation $\tau \in S_n$, $R_\sigma : \ket{\sigma} \mapsto \ket{\sigma \circ \tau^{-1}}$ (where $\sigma \circ \tau^{-1}$ is the composition of $\tau$ with $\sigma$ from the right). Clearly, elements of $\mathcal{A}$ and $\mathcal{E}$ commute with each other. It turns out that they are furthermore the entire commutants of each other.

Under the action of $\mathcal{A}$ and $\mathcal{E}$ we can decompose $\mathbb{C} [S_n] \cong \bigoplus_\lambda V_\lambda \otimes W_\lambda$. It turns out that the dimensions of the left and right representaitons are equal, i.e., $\dim (V_\lambda) = \dim (W_\lambda) =: d_\lambda$. These dimensions satisfy $\sum_\lambda d_\lambda^2 = n!$. Furthermore, the basic (identity) input state $\ket{12\ldots n}$ is equally supported everywhere, i.e.,
$$
    \ket{12\ldots n} = \sum_\lambda \frac{d_\lambda}{\sqrt{n!}} \ket{\phi_\lambda}
$$
where $\ket{\phi_\lambda}$ is the maximally entangled state along $V_\lambda \otimes W_\lambda$. 

Now, suppose we start with $\ket{12\ldots n}$, apply a permutational circuit $C$, and measure the permutational basis. The amplitude corresponding to measuring the identity permutation is equal to 
\begin{align}
    \begin{split}
        \bra{12\ldots n} C \ket{12 \ldots n} &= \sum_\lambda \frac{d^2_\lambda}{{n!}} \bra{\phi_\lambda} c_\lambda \otimes I \ket{\phi_\lambda} \\
        &= \sum_\lambda \frac{d^2_\lambda}{{n!}} \frac{Tr(c_\lambda)}{d_\lambda}
    \end{split}
\end{align}
where $C \cong \bigoplus_\lambda c_\lambda \otimes I$. If we can evaluate each normalized trace $\frac{Tr(c_\lambda)}{d_\lambda}$ within additive error, we can estimate the amplitude. Using the observation that normalized trace is a complete problem for the one-clean-qubit model\textbf{DQC}1, \cite{aaronson2017computational} proved that individual amplitudes of $C$ in the permutational basis can be evaluated using \textbf{DQC}1. The main difficulty in establishing this result is that the permutational basis is not binary, and the \textbf{DQC}1 cannot afford more than logarithmic ancilla bits. A concise representation of permutational bases was used to establish this result. They furthermore showed that one can perform universal quantum computation if we start from specific irreducible representations.

In summary, the power of the ball-permuting model significantly depends on the input from which the model starts and is measured. If the model starts from permutational bases, we get a weak model where individual amplitudes can be estimated within \textbf{DQC}1. If we begin from other input states (corresponding to irreducible representations of the symmetric group), we get \textbf{BQP}. \cite{aaronson2017computational} asked the following question: ``Are there specific inputs which give us a model that is intermediate between \textbf{DQC}1 and \textbf{BQP}?'' For that purpose they introduced \textbf{SampTQP} which we call $\frac12$\textbf{BQP}. In other words, $\frac12$\textbf{BQP} is equivalent to the power of ball permutational model where we start from $\ket{\phi_\lambda}$ and measure the quantum state in the basis of $V_\lambda \otimes W_\lambda$. We conjecture that a similar scenario holds for any gate set model, which allows many conserved quantities.

\subsection{A complete problem for $\frac12$\textbf{BQP}}
\label{subsec:half-bqp-complete-problem}
We now give a complete problem for $\frac12$\textbf{BQP}, based on estimating the expectation value of the post-processing function under the distribution implied by the $\frac12$\textbf{BQP} quantum circuit. The motivation for defining this problem is to have a closed-form expression that captures the power of $\frac12$\textbf{BQP}, analogous to the completeness of the normalized trace for \textbf{DQC}1. As we will see, Fourier analysis on this closed-form expression will allow us to relate $\frac12$\textbf{BQP} and \textbf{DQC}1. In particular, \textbf{DQC}1 can simulate a $\frac12$\textbf{BQP} computation when the $\frac12$\textbf{BQP} classical post-processing has a learnable Fourier spectrum. Beyond this result, our hope is that having a closed-form expression may help with proving $\frac12$\textbf{BQP} lower bounds going forward.

\begin{theorem}
  \label{thm:half-bqp-complete-problem}
  Let $f: \{0,1\}^n \times \{0,1\}^n \rightarrow [-1,1]$ be computable by a polynomial-time Turing machine and
  let $C$ be a polynomial-size quantum circuit.
  Then it is promise-$\frac12$\textbf{BQP} complete to estimate the following quantity within inverse polynomial additive error:
  \[A = \frac{1}{2^n}\sum_{w,z \in \{0,1^n\}} \Big| \bra{z}C\ket{w} \Big|^2 f(w,z). \]
\end{theorem}
Note that as with the weighted trace, for this sum to be meaningful,
$f$ must be non-zero on a polynomial fraction of input pairs.
\begin{proof}
    To see containment in $\frac12$\textbf{BQP}, notice that one run of a $\frac12$\textbf{BQP} circuit outputs $f(w,z)$ with probability
    \[\Pr(w,z) = \frac{1}{2^n} \Bigl| \bra{z}C\ket{w} \Bigr|^2.\]
    By a standard Chernoff bound, we can estimate $A$ within additive error using polynomially many samples from $C$ which can be run in parallel on our $\frac12$\textbf{BQP} machine.
     
    For the completeness proof, let $f$ be the classical post-processing function which is $1$ if the $\frac12$\textbf{BQP} circuit accepts and $0$ if $\frac12$\textbf{BQP} rejects. Then $A$ is exactly the acceptance probability of the $\frac12$\textbf{BQP} machine, and estimating $A$ within additive error is sufficient to determine whether to accept or reject.
     
    As a small caveat, note that $f$ is deterministic, while the post-processing may have randomness. This randomness can be simulated by including more registers in the $\frac12$\textbf{BQP} circuit and using this randomness as part of the input to $f$.
\end{proof}

\subsection{When can \textbf{DQC}1 simulate $\frac12$\textbf{BQP}?}
\label{subsec:when-dqc1-simulates-half-bqp}

So far we gave evidence that $\frac12$\textbf{BQP} is probably more powerful than \textbf{DQC}1. To understand the boundary between these two complexity classes, in this Section, we study the limitations one needs to impose on the definition of $\frac12$\textbf{BQP} in order to recover \textbf{DQC}1. In particular we show that there is a \textbf{DQC}1 simulation of $\frac12$\textbf{BQP} if the classical post-processing function has a learnable Fourier spectrum (see \cref{thm:when-dqc1-simulates-half-bqp} below).

 In order to achieve this we use a Fourier representation of $\frac12$\textbf{BQP} computation. As we have seen, the problem of estimating
\begin{align*}
    A = \frac{1}{2^n}\sum_{w,z} \Bigl|\bra{z}C\ket{w}\Bigr|^2 f(w,z)
\end{align*}
is $\frac12$\textbf{BQP}-complete. If we write $f$ using its Fourier decomposition, we obtain a weighted sum of normalized traces of unitaries, which can be estimated in \textbf{DQC}1. This suggests that we can study the boundary between \textbf{DQC}1 and $\frac12$\textbf{BQP} in terms of Fourier coefficients. The analysis is as follows:
\begin{align}
\begin{split}
    A &= \frac{1}{2^n}\sum_{w,z} \Bigl|\bra{z}C\ket{w}\Bigr|^2 f(w,z) \\
    &= \frac{1}{2^n}\sum_{w,z} \Tr(\ketbra{w}{w}C^\dagger \ketbra{z}{z} C) f(w,z) \\
    &= \frac{1}{2^n}\sum_{w,z} \Tr(\ketbra{w}{w}C^\dagger \ketbra{z}{z} C) \sum_{s,s'} \wh{f}_{s,s'} (-1)^{s\cdot w + s' \cdot z} \\
    &= \frac{1}{2^n}\sum_{s,s'} \wh{f}_{s,s'} \sum_{w,z} (-1)^{s\cdot w + s' \cdot z} \Tr(\ketbra{w}{w}C^\dagger \ketbra{z}{z} C) \\
    &= \frac{1}{2^n}\sum_{s,s'} \wh{f}_{s,s'} \Tr(\sum_{w,z} (-1)^{s\cdot w + s' \cdot z} \ketbra{w}{w}C^\dagger \ketbra{z}{z} C) \\ 
    &= \frac{1}{2^n}\sum_{s,s'} \wh{f}_{s,s'} \Tr(\Bigl(\sum_w (-1)^{s\cdot w}\ketbra{w}{w}\Bigr)C^\dagger \Bigl(\sum_z (-1)^{s' \cdot z} \ketbra{z}{z}\Bigr) C) \\
    &= \frac{1}{2^n}\sum_{s,s'} \wh{f}_{s,s'} \Tr(Z^sC^\dagger Z^{s'}C) \\
    &= \sum_{s,s'} \wh{f}_{s,s'} \frac{1}{2^n}\Tr(Z^sC^\dagger Z^{s'}C)
\end{split}
\end{align}
Thus it is $\frac12$\textbf{BQP}-complete to estimate 
\begin{align*}
    A = \sum_{s,s'} \wh{f}_{s,s'} \frac{1}{2^n}\Tr(Z^sC^\dagger Z^{s'}C).
\end{align*}
We call this the \textit{Fourier representation} of a $\frac12$\textbf{BQP} computation.
Note that for fixed $s,s'$, $\frac{1}{2^n}\Tr(Z^sC^\dagger Z^{s'}C)$ is just the normalized trace of some unitary $V_{s,s'} = Z^sC^\dagger Z^{s'}C$. This suggests a relationship between \textbf{DQC}1 and $\frac12$\textbf{BQP}. In particular,
as a corollary to the $\frac12$\textbf{BQP}-completeness of the Fourier representation, any $\frac12$\textbf{BQP} computation such that $f(w,z)$ is given by a sufficiently sparse representation in the Fourier basis is computable in \textbf{DQC}1. This is captured by the following theorem:
\begin{theorem}
    \label{thm:when-dqc1-simulates-half-bqp}
    Let $C$ be a polynomial-size quantum circuit. Let $f:\{0,1\}^n \times \{0,1\}^n \rightarrow [-1,1]$ be a polynomial-time computable function such that there exists $M \subset \{0,1\}^{2n}$ with $|M| \leq poly(n)$ satisfying
    \[\sum_{\substack{(s,s') \in \{0,1\}^n \times \{0,1\}^n \\ (s,s') \notin M}} |\wh{f}_{s,s'}| \leq \varepsilon\]
    where $\varepsilon = \frac{1}{poly(n)}$.
    Then the quantity
    \begin{align*}
        A = \frac{1}{2^n}\sum_{w,z \in \{0,1\}^n} \Bigl|\bra{z}C\ket{w}\Bigr|^2 f(w,z)
    \end{align*}
    can be estimated within $\pm (\varepsilon + \frac{1}{poly(n)})$ by a machine which can run polynomially many \textbf{DQC}1 circuits in parallel with \textbf{BPP} pre and post-processing.
\end{theorem}

\begin{proof}
    First, note that every $(s,s') \notin M$ must have $|\wh{f}_{s,s'}| \leq \varepsilon$. Thus, by \cref{fact:learning-fourier} and \cref{fact:goldreich-levin}, there is a polynomial-time algorithm to estimate all Fourier coefficients in $M$ within inverse polynomial additive error. Furthermore, we can use our \textbf{DQC}1 circuits to estimate 
    \[D_{s,s'} = \frac{1}{2^n}\Tr(Z^sC^\dagger Z^{s'}C)\]
    within inverse polynomial additive error for all $(s,s') \in M$.
    Let $\widetilde{f_{s,s'}}$ and $\widetilde{D_{s,s'}}$ be our estimates for the Fourier coefficients and normalized traces, respectively. Then a straightforward error analysis shows that computing the sum 
    \[\sum_{(s,s') \in M} \widetilde{{f}_{s,s'}} \widetilde{D_{s,s'}}\]
    yields an estimate of $A$ within $\pm (\varepsilon + \frac{1}{poly(n)})$.
\end{proof}
Intuitively, the theorem says that \textbf{DQC}1 can simulate a $\frac12$\textbf{BQP} computation when $f$ is sufficiently sparse in the Fourier basis.

\subsection{$\frac12$\textbf{BQP} can simulate \textbf{IQP}}
\label{subsec:half-bqp-simulates-iqp}

We now show that $\frac12$\textbf{BQP} can simulate \textbf{IQP}.
\begin{theorem}
    \label{thm:half-bqp-simulates-iqp}
    $\frac12$\textbf{BQP} can sample from the output distribution of \textbf{IQP} machines.
\end{theorem}
\begin{proof}
    Let $C = H^{\otimes n}DH^{\otimes n}$ be an \textbf{IQP} circuit. Note that $D$ is diagonal in the Pauli $Z$ basis.
    Now consider the following $\frac12$\textbf{BQP} circuit:
    
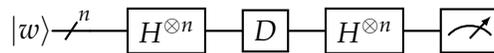
\begin{figure}[H]
        \centering
        \begin{quantikz}
            \ket{w} & \qwbundle{n} & \gate{H^{\otimes n}} & \gate{D} & \gate{H^{\otimes n}} & \meter{}
        \end{quantikz}
    \end{figure}
    Equivalently, we can consider
    
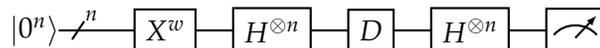
\begin{figure}[H]
        \centering
        \begin{quantikz}
            \ket{0^n} & \qwbundle{n} & \gate{X^w} & \gate{H^{\otimes n}} & \gate{D} & \gate{H^{\otimes n}} & \meter{}
        \end{quantikz}
    \end{figure}
    Taking advantage of commutation relations yields the following sequence of equivalent circuits:
    \begin{figure}[H]
        \centering
        \;\;
        \begin{quantikz}
            \ket{0^n} & \qwbundle{n} & \gate{X^w} & \gate{H^{\otimes n}} & \gate{D} & \gate{H^{\otimes n}} & \meter{}
        \end{quantikz} \\
        =
        \begin{quantikz}
            \ket{0^n} & \qwbundle{n} & \gate{H^{\otimes n}} & \gate{Z^w} & \gate{D} & \gate{H^{\otimes n}} & \meter{}
        \end{quantikz} \\
        =
        \begin{quantikz}
            \ket{0^n} & \qwbundle{n} & \gate{H^{\otimes n}} & \gate{D} & \gate{Z^w} & \gate{H^{\otimes n}} & \meter{}
        \end{quantikz} \\
        =
        \begin{quantikz}
            \ket{0^n} & \qwbundle{n} & \gate{H^{\otimes n}} & \gate{D} & \gate{H^{\otimes n}} & \gate{X^w} & \meter{}
        \end{quantikz}
    \end{figure}
    which is equivalent to the original \textbf{IQP} computation except the output is shifted by $w$. Thus, to simulate \textbf{IQP}, we sample from this circuit and shift the measured string by $w$.
\end{proof}
Note that we could have simply used the fact that an \textbf{IQP} circuit $C$ is diagonal in the Pauli $X$ basis to commute $X^w$ to the end. However, we think the given proof is a nice demonstration of how ``passing'' $X^w$ through the circuit can be useful.
 
Having seen that $\frac12$\textbf{BQP} can simulate both \textbf{DQC}1 and \textbf{IQP}, it is clear that $\frac12$\textbf{BQP} is at least as powerful as these models. However, it is important to ask whether this inclusion is strict. While this question remains unresolved, we will see that $\frac12$\textbf{BQP} solves many standard oracle problems that are not believed to be in \textbf{DQC}1 or \textbf{IQP}.


\section{Separating $\frac12$\textbf{BQP} and \textbf{BPP}}
\label{sec:bqp-vs-bpp}

In this section we show that $\frac12$\textbf{BQP} solves many of the usual oracle problems used to demonstrate quantum advantage over classical computation. In particular, $\frac12$\textbf{BQP} can solve the Deutsch-Jozsa, Bernstein-Vazirani, Simon's, and Period Finding problems (in the black box setting). We only give proofs for Simon's problem and Period Finding, since the proofs for Deutsch-Jozsa and Bernstein-Vazirani are essentially identical to the proof for Simon's problem. We end this Section by showing that \textsc{Order Finding} and \textsc{Factoring} can be implemented in $\frac12$\textbf{BQP}.

Since all these problems (besides Deutsch-Jozsa) are not believed to be in \textbf{DQC}1 or \textbf{IQP}, these results strongly suggest that $\frac12$\textbf{BQP} captures a strictly stronger complexity class. Furthermore, since these separations provide some of the strongest known separations between \textbf{BQP} and \textbf{BPP}, they suggest that in some sense, $\frac12$\textbf{BQP}, in terms of algorithmic capabilities, might not be very far from \textbf{BQP}. However, as we show in \cref{subsec:bqp-separation}, there is an oracle separation between $\frac12$\textbf{BQP} and \textbf{BQP}, showing that $\frac12$\textbf{BQP} suffers from some of the same limitations as \textbf{DQC}1. We furthermore show that $\frac{1}{2}$\textbf{BQP} cannot obtain the quadratic speed up given by Grover's algorithm (See \cref{thm:grover}). 
All together, these results evidence the status of $\frac12$\textbf{BQP} as a nontrivial intermediate model of quantum computation.

\subsection{Fourier Sampling}
\label{subsubsec:fourier-sampling}
\textsc{Fourier Sampling} is defined in \cite{aaronson2014forrelation}. We now show that it gives a sampling query complexity separation between \textbf{BPP} and $\frac12$\textbf{BQP}.

\begin{definition}[\textsc{Fourier Sampling,aaronson2014forrelation}]
    Given oracle access to a boolean function $f: \{0,1\}^n \rightarrow \{-1,1\}$, sample from a distribution 
    $D$ over $\{0,1\}^n$ such that $\norm{D - D_f} \leq \varepsilon$ where $D_f$ is the distribution defined by 
    \begin{align*}
        \Pr[y] = \hat{f}(y)^2 = \Biggl(\frac{1}{2^n} \sum_{x \in \{0,1\}^n} (-1)^{x\cdot y}f(x)\Biggr)^2
    \end{align*}
\end{definition}
\textsc{Fourier Sampling} is solved for $\varepsilon = 0$ in \textbf{BQP} using $1$ query by first applying $H^{\otimes n}$, then the function oracle, then $H^{\otimes n}$ again and measuring. 
Furthermore, there is a theorem due to \cite{aaronson2014forrelation} that for small constant $\varepsilon$, the randomized classical query complexity of \textsc{Fourier Sampling} is $\Omega(2^n/n)$.
We now give the result:

\begin{theorem}
    \textsc{Fourier Sampling} is solvable in $\frac12$\textbf{BQP} using $1$ query.
    \label{thm:fourier-sampling}
\end{theorem}

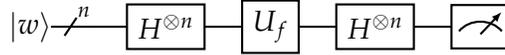
\begin{figure}[H]
        \centering
        \begin{quantikz} 
            \ket{w} & \qwbundle{n} & \gate{H^{\otimes n}} & \gate{U_f} & \gate{H^{\otimes n}} & \meter{}
        \end{quantikz}
        \caption{A $\frac12$\textbf{BQP} circuit for \textsc{Fourier Sampling}}
        \label{fig:fourier-sampling}
\end{figure}
\begin{proof}
    The circuit is given in \cref{fig:fourier-sampling}. After measuring a label $y$ and learning $w$, the algorithm outputs $y \oplus w$.
    
    The state at the end of the computation is given by
    \begin{align}
    \begin{split}
        & \frac{1}{2^n}\sum_{x,y \in \{0,1\}^n}(-1)^{w\cdot x + x\cdot y} f(x) \ket{y} \\
        =& \frac{1}{2^n}\sum_{y \in \{0,1\}^n}\sum_{x \in \{0,1\}^n}(-1)^{x \cdot (w \oplus y)} f(x) \ket{y} 
    \end{split}
    \end{align}
    and the probability of measuring a label $y$ is
    \begin{align}
    \begin{split}
        \Pr[y] &= \Biggl(\frac{1}{2^n}\sum_{x \in \{0,1\}^n}(-1)^{x \cdot (w \oplus y)} f(x) \Biggr)^2 \\
        &= \wh{f}(y \oplus w)^2.
    \end{split}
    \end{align}
    Now let $y' = y \oplus w$ be the output of the algorithm after sampling $y$ and learning $w$. Then the algorithm outputs $y'$ with probability $\wh{f}(y')^2$. Thus, \textsc{Fourier Sampling} is solvable in $\frac12$\textbf{BQP} using $1$ quantum query.
\end{proof}
Note that we could have used the proof technique from \cref{thm:half-bqp-simulates-iqp} because $U_f$ is diagonal in the Pauli $Z$ basis.

\subsection{Simon's Problem}
\label{subsubsec:simons-problem}
Simon's problem was defined in \cite{simon1997power} to give an oracle separation between \textbf{BQP} and \textbf{BPP}. In this section we show that Simon's problem is solvable by $\frac12$\textbf{BQP}, and therefore gives an oracle separation between $\frac12$\textbf{BQP} and \textbf{BPP}. Since \textbf{DQC}1 is not known to contain Simon's problem, this result also provides evidence that $\frac12$\textbf{BQP} is strictly above \textbf{DQC}1.
\begin{definition}[Simon's problem]
    Given oracle access to a function $f:\{0,1\}^n \rightarrow \{0,1\}^n$ with the promise that $\exists s \in \{0,1\}^n$ such that $\forall x,y \in \{0,1\}^n$, \[f(x) = f(y) \iff x \oplus y \in \{0^n,s\},\]
    find $s$.
\end{definition}
\begin{theorem}
    \label{thm:simons-problem}
    Simon's problem is solvable in $\frac12$\textbf{BQP}
\end{theorem}
\begin{proof}
    The algorithm is the same as Simon's algorithm, but we apply a bitwise xor to the output in classical post-processing as we did for Fourier Sampling. We start with a uniformly random computational basis state $\ket{w}$ and for convenience we denote $\ket{w} = \ket{w_1}\ket{w_2}$.
    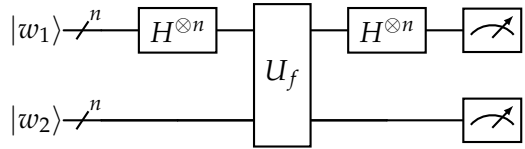
\begin{figure}[H]
        \centering
        \begin{quantikz}
            \ket{w_1} & \qwbundle{n} & \gate{H^{\otimes n}} & \gate[2]{U_f} & \gate{H^{\otimes n}} & \meter{} \\
            \ket{w_2} & \qwbundle{n} & \qw & \qw & \qw & \meter{}
        \end{quantikz}
        \caption{a circuit for Simon's algorithm}
    \end{figure}
    The analysis is as follows:
    \begin{align}
    \begin{split}
        \ket{w_1}\ket{w_2} \xmapsto{H^{\otimes n}}& \frac{1}{\sqrt{2^n}}\sum_x (-1)^{w_1 \cdot x} \ket{x}\ket{w_2} \\
        \xmapsto{U_f}& \frac{1}{\sqrt{2^n}}\sum_x (-1)^{w_1 \cdot x} \ket{x}\ket{f(x) \oplus w_2}
        \\
        \xmapsto{H^{\otimes n}}& \frac{1}{2^n}\sum_x \sum_y (-1)^{w_1 \cdot x}(-1)^{x \cdot y} \ket{y}\ket{f(x) \oplus w_2} \\
        =& \frac{1}{2^n}\sum_y \sum_x (-1)^{x \cdot (w_1 \oplus y)} \ket{y}\ket{f(x) \oplus w_2}
    \end{split}
    \end{align}
    Now the probability of measuring a label $y$ is given by
    \begin{align}
        \Pr[y] = \frac{1}{2^{2n}} \Biggl|\sum_x (-1)^{x \cdot (w_1 \oplus y)}\ket{f(x) \oplus w_2}\Biggr|^2
    \end{align}
    For $z \in \Ima(f)$ let $a_z \neq b_z$ be the inputs such that $f(a_z) = f(b_z) = z$. Note that $a_z \oplus b_z = s$.
    Now
    \begin{align}
    \begin{split}
        \Pr[y] &= \frac{1}{2^{2n}} \Biggl|\sum_{z \in \Ima(f)} \Bigl[ (-1)^{a_z\cdot(w_1 \oplus y)} + (-1)^{b_z\cdot(w_1 \oplus y)} \Bigr] \ket{z \oplus w_2} \Biggr|^2 \\
        &= \frac{1}{2^{2n}} \sum_{z \in \Ima(f)} \Bigl[ (-1)^{a_z\cdot(w_1 \oplus y)} + (-1)^{b_z\cdot(w_1 \oplus y)} \Bigr]^2 \\
        &= \frac{1}{2^{2n}} \sum_{z \in \Ima(f)} \Bigl[ (-1)^{a_z\cdot(w_1 \oplus y)} + (-1)^{(a_z \oplus s)\cdot(w_1 \oplus y)} \Bigr]^2 \\
        &= \frac{1}{2^{2n}} \sum_{z \in \Ima(f)} \Bigl[ (-1)^{a_z\cdot(w_1 \oplus y)} + (-1)^{a_z\cdot(w_1 \oplus y)}(-1)^{s\cdot(w_1 \oplus y)} \Bigr]^2\\
        &= \frac{1}{2^{2n}} \sum_{z \in \Ima(f)} \Bigl[ (-1)^{a_z\cdot(w_1 \oplus y)}(1 + (-1)^{s \cdot (w_1 \oplus y)}) \Bigr]^2 \\
    \end{split}
    \end{align}
    Analyzing the last equation, when $s \cdot (w_1 \oplus y) = 1$, $\Pr[y] = 0$.
    Thus, with probability $1$ we measure $y$ such that $s \cdot (w_1 \oplus y) = 0$.
    Since we learn $w_1$ and $y$, we can solve the corresponding system of equations to find $s$ using $O(n)$ queries.
\end{proof}

We note that the tools used above are immediately applicable to the Bernstein-Vazirani and Deustch-Jozsa problems. Due to the similarity of ideas we omit the proofs.

\begin{theorem}
    \label{thm:dj-bv}
    $\frac{1}{2}$\textbf{BQP} can solve the Deustch-Jozsa and Bernstein-Vazirani problems. 
\end{theorem}

\subsection{Period finding}
\label{subsubsec:period-finding}
In this section, we show that $\frac12$\textbf{BQP} can solve \textsc{Period Finding}. As with Simon's problem, \textsc{Period Finding} is not known to be in \textbf{DQC}1. Furthermore, it is not known to be in \textbf{NISQ}. Thus, this result provides further evidence that $\frac12$\textbf{BQP} is strictly more powerful than \textbf{DQC}1 and also that $\frac12$\textbf{BQP} may be incomparable to \textbf{NISQ}.
\textsc{Period Finding,shor1994algorithms} is defined as follows.

\begin{definition}[\textsc{Period Finding}]
    Given oracle access to a function $f:[N] \rightarrow \{0,1\}^n$ with the promise that $\exists r \in [N]$ with $r \leq \sqrt{N}$ such that $\forall x,y \in [N]$, \[f(x) = f(y) \iff x \equiv y \mod r,\]
    find $r$.
\end{definition}

\begin{theorem}
     \label{thm:period-finding}
    \textsc{Period Finding} is in $\frac12$\textbf{BQP}.
\end{theorem}
\begin{proof}
    We use the same circuit as Shor's original period finding algorithm. Here $n = \lceil \log N \rceil$,  $w_1$ and $w_2$ are uniformly random bitstrings, and $Q = 2^{n}$.  
    \begin{figure}[H]
        \centering
        \begin{quantikz}
            \ket{w_1} & \qwbundle{n} & \gate{QFT} & \gate[2]{U_f} & \gate{QFT} & \meter{} \\
            \ket{w_2} & \qwbundle{n} & \qw & \qw & \qw & \meter{}
        \end{quantikz}
        \caption{a $\frac12$\textbf{BQP} circuit for \textsc{Period Finding}}
    \end{figure}
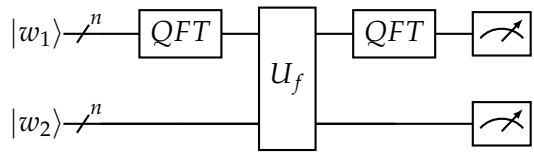
    The analysis is as follows:
    \begin{align}
    \begin{split}
        \ket{w_1}\ket{w_2} \xmapsto{QFT}& \frac{1}{\sqrt{Q}}\sum_{x=0}^{Q-1} \omega_Q^{w_1x} \ket{x}\ket{w_2} \\
        \xmapsto{U_f}& \frac{1}{\sqrt{Q}}\sum_{x=0}^{Q-1} \omega_Q^{w_1x} \ket{x}\ket{f(x) \oplus w_2} \\
        \xmapsto{QFT}& \frac{1}{Q}\sum_{x=0}^{Q-1} \sum_{y=0}^{Q-1} \omega_Q^{w_1x}\omega_Q^{xy} \ket{y}\ket{f(x) \oplus w_2} \\
        =& \frac{1}{Q}\sum_{y=0}^{Q-1} \sum_{x=0}^{Q-1} \omega_Q^{x(w_1 + y)} \ket{y}\ket{f(x) \oplus w_2}
    \end{split}
    \end{align}
    Now the probability of measuring a label $y$ is given by
    \begin{align}
    \begin{split}
        \Pr[y] &= \frac{1}{Q^2} \Biggl| \sum_{x=0}^{Q-1} \omega_Q^{x(w_1 + y)} \ket{f(x) \oplus w_2} \Biggr|^2 \\
        &= \frac{1}{Q^2} \Biggl| \sum_{z \in \Ima(f)} \sum_{x \in f^{-1}(z)} \omega_Q^{x(w_1 + y)} \ket{z \oplus w_2} \Biggr|^2 \\
        &= \frac{1}{Q^2} \sum_{z \in \Ima(f)} \Biggl| \sum_{x \in f^{-1}(z)} \omega_Q^{x(w_1 + y)} \Biggr|^2 \\
        &= \frac{1}{Q^2} \sum_{z \in \Ima(f)} \Biggl| \sum_{j = 0}^{L-1} \omega_Q^{(x_0 + jr) (w_1 + y)} \Biggr|^2 \\
        &= \frac{1}{Q^2} \sum_{z \in \Ima(f)} \Biggl| \sum_{j = 0}^{L-1} \omega_Q^{jr(w_1 + y)} \Biggr|^2
    \end{split}
    \end{align}
    Using the standard analysis of Shor's algorithm, \[\Pr[w_1 + y = jQ/r] \geq \frac{4}{\pi^2}.\]
    Thus, we can find $r$ by computing $w_1 + y \mod Q$ and using the continued fraction procedure.
\end{proof}
We have shown that in the oracle setting, $\frac12$\textbf{BQP} can solve \textsc{Period Finding}. The next question to ask is whether $\frac12$\textbf{BQP} can implement Shor's factoring algorithm (outside of the oracle setting), which reduces \textsc{Factoring} to finding the period of the modular exponentiation function (i.e., order finding). While it is not immediately obvious how to implement the modular exponentiation step of Shor's algorithm reversibly using few clean ancillas, we show that it is possible, due to several nontrivial results about logarithmic-depth circuits. 

\subsection{Order finding and factoring}
\label{subsubsec:factoring}
In this section, we show that $\frac12$\textbf{BQP} solves the \textsc{Order Finding} problem, and therefore \textsc{Factoring}. The reduction from \textsc{Factoring} to \textsc{Order Finding} is standard, so here we just show that \textsc{Order Finding} is in $\frac12$\textbf{BQP}. A helpful review of the reduction can be found in \cite{cleve2000fast} or any textbook on quantum computation.
\begin{definition}[\textsc{Order Finding}]
    Given positive integers $N$ and $a < N$ such that $\gcd(a,N) = 1$, find the least positive integer $r$ such that $a^r = 1 \mod N$. 
\end{definition}
In the language of group theory, \textsc{Order Finding} corresponds to finding the order of $a$ in the multiplicative group of integers mod $N$.
Note that \textsc{Order Finding} is a special case of \textsc{Period Finding} (\cref{subsubsec:period-finding}). In particular, it is the same as finding the period of the function 
\[f_a(x) = a^x \mod N.\]

\begin{theorem}
    \label{thm:order-finding}
    \textsc{Order Finding} is in $\frac12$\textbf{BQP}.
\end{theorem}

\begin{proof}
    Since we have already shown that $\frac12$\textbf{BQP} solves \textsc{Period Finding} in the oracle setting, to show that $\frac12$\textbf{BQP} solves \textsc{Order Finding}, we only need to show how to implement $f_a$ within a $\frac12$\textbf{BQP} circuit. Concretely if we can implement
\begin{align*}
    U_f: \ket{x}\ket{w} \mapsto \ket{x}\ket{(a^x \mod N) \oplus w}
\end{align*}
as a reversible circuit using $O(\log(\log N))$ clean ancillae, where $x \in \mathbb{N}$ is represented by $n=2\log N$ bits, and $w \in \{0,1\}^n$, then we can find the order of $a$ in the multiplicative group of integers mod $N$.

\begin{claim}
    There is a polynomial-size reversible classical circuit which uses a constant number of workspace ancillae and implements the transformation
    \begin{align*}
        U_f: \ket{x}\ket{w} \mapsto \ket{x}\ket{(a^x \mod N) \oplus w}.
    \end{align*}
\end{claim}

At a high level, the proof of the claim is the following:
First, for a given $a,N$ with $a < N$, there is a polynomial-time-generated \textbf{NC1} circuit which computes $a^x \mod N$. Now, due to Barrington's theorem \cite{barrington1986bounded}, a logarithmic-depth circuit can be simulated by a polynomial-depth, width-$5$ permutation branching program. Furthermore, due to a result by Ambainis \cite{ambainis2000computing}, a width-$5$ permutation branching program can be simulated by a reversible circuit with $3$ workspace ancillae.
 
As an important note, the permutation branching program and the reversible circuit with $3$ workspace ancillae do not output the entire representation of $a^x \mod N$, rather they enable us to recognize a given language in \textbf{NC1}. Thus, to output $a^x \mod N$ into the second register, we simulate $n$ copies of the $NC1$ circuit, each recognizing the language corresponding to a single bit of $a^x \mod N$.
 
Now, it remains to show that there exists a polynomial-time generated logarithmic-depth (\textbf{NC1}) circuit which computes $a^x \mod N$. We will need several facts about \textbf{NC1} circuits.

\begin{fact}
    \label{fact:nc1-addition}
    There exist uniformly-generated boolean (\textbf{NC1}) circuit families  with depth $O(\log n)$ for addition and subtraction of two $n$-bit integers.
\end{fact}

\begin{fact}[\cite{beame1986log,chiu2000nc1}]
    \label{fact:nc1-iterated-product}
    There exists a uniformly-generated boolean (\textbf{NC1}) circuit family with depth $O(\log n)$ which computes the product
    \[\prod_{i=1}^n x_n\]
    where each $x_i$ is an $n$-bit integer.
\end{fact}

\begin{fact}[\cite{beame1986log,chiu2000nc1}]
    \label{fact:nc1-division}
    There exists a uniformly-generated boolean (\textbf{NC1}) circuit family with depth $O(\log n)$ which computes the quotient $\lfloor \frac{x}{y} \rfloor$
    where $x$ and $y$ are $n$-bit integers.
\end{fact}

Taken together, these facts say that standard arithmetic operations are in \textbf{NC1}, as well as computing the product of $n$ integers. 
 
Now, we want to construct a depth $O(\log n)$ fan-in $2$ boolean circuit which computes $a^x \mod N$. The following construction is due to Cleve and Watrous \cite{cleve2000fast}. Let $x = x_{2n-1}x_{2n-2}\ldots x_1x_0$. Then 
\begin{align}
        a^x = \prod_{i=0}^{2n-1} (a^{2^i})^{x_i}.
\end{align}
Since $a$ is known ahead of time (i.e. the classical preprocessing step of $\frac12$\textbf{BQP} has access to $a$), we can precompute $b_i = a^{2^i} \mod N$ for all $0 \leq i \leq 2n-1$, where each $b_i$ is an $n$-bit integer. Thus, by \cref{fact:nc1-iterated-product} there is a $O(\log n)$-depth circuit which computes
\[A = \prod_{i=0}^{2n-1} b_i^{x_i}.\]
Note that $A$ is a $2n^2$-bit integer and $A \equiv a^x \mod N$. Now, by \cref{fact:nc1-addition} and \cref{fact:nc1-division} there is a $O(\log n)$-depth circuit which computes 
\[R = A - N\lfloor \frac{A}{N} \rfloor.\]
Note that $R$ is the remainder when dividing $A$ by $N$, and therefore 
\[R = a^x \mod N.\]
Also note that each of $A,R,N$ and $b_i$ is at most $2n^2$ bits. Therefore all of the operations can be computed by $O(\log n)$-depth circuits. Since the whole circuit is a composition of a constant number of $O(\log n)$-depth circuits, the entire circuit has depth $O(\log n)$. Thus, there is a polynomial-time generated depth $O(\log n)$ fan-in 2 boolean circuit which computes $a^x \mod N$.
\end{proof}

\begin{corollary}
    \textsc{Factoring} is in $\frac12$\textbf{BQP}.
\end{corollary}
This corollary follows by the standard reduction from \textsc{Factoring} to \textsc{Order Finding}. 
This result is important because it provides an unrelativized separation between $\frac12$\textbf{BQP} and \textbf{BPP} under a standard cryptographic assumption (\textsc{Factoring} is not in \textbf{BPP}). That is, even if we adpot a skeptical perspective on the usefulness of oracle separations, there is still compelling evidence that $\frac12$\textbf{BQP} captures a more powerful model of computation.

\subsubsection{Can $\frac12$\textbf{BQP} solve the Abelian HSP?}
\label{subsubsec:hsp}

We have seen that $\frac12$\textbf{BQP} can solve many hidden subgroup problems in the oracle setting, as well as \textsc{Order Finding} and \textsc{Factoring} outside of the oracle setting. A natural next question is whether $\frac12$\textbf{BQP} can solve the general Abelian hidden subgroup problem (HSP). The short answer is that we don't know. To see why known quantum algorithms do not easily carry over to $\frac12$\textbf{BQP}, we now give a brief overview of the quantum algorithm for the Abelian HSP.

\begin{definition}[Abelian HSP]
    Let $G$ be a finite Abelian group, $X$ a finite set, and $f: G \rightarrow X$ a function with the promise that there exists a subgroup $H \leqslant G$ such that
    \[f(x) = f(y) \iff x - y \in H.\]
    Given an index oracle $U_f$ implementing $f$, find a generating set for $H$.
\end{definition}

The quantum algorithm for Abelian HSP consists of the following steps. Note that Period Finding and Simon's problem are special cases of this algorithm.

\begin{enumerate}
    \item Create a uniform superposition over all group elements.
    \item Query the oracle $U_f$ in superposition and store the contents in a new register.
    \item Apply the quantum Fourier transform over the group $G$ to the first register, then measure.
    \item Repeat. Each measurement yields information which decreases the size of possible $H$ by a factor of at least $2$.\footnote{The details of this post-processing step are unimportant for our discussion and can be found in any standard textbook on quantum computation} Thus, after $O(\log|G|)$ runs, we obtain $H$.
\end{enumerate}

The crucial difference between the general Abelian HSP and the special cases we have seen so far appears in step $3$. In particular, implementing the Fourier transform over an arbitrary group $G$ involves decomposing the group $G$ into a product of cyclic groups, $G = \mathbb{Z}_{N_1} \times \mathbb{Z}_{N_2} \times \ldots \times \mathbb{Z}_{N_m}$ then performing the quantum Fourier transform over each cyclic group in parallel \cite{cheung2001decomposing}. To decompose $G$, we use the quantum algorithm for \textsc{Order Finding} to find the factorization of $G$, which can be performed by a $\frac12$\textbf{BQP} computation. Unfortunately, the circuit for the quantum Fourier transform over arbitrary cyclic groups \cite{kitaev1995quantum} relies on quantum phase estimation, an algorithmic primitive which we do not currently know how to implement in $\frac12$\textbf{BQP}. As a result, a $\frac12$\textbf{BQP} implementation of the quantum Fourier transform over arbitrary cyclic groups would enable $\frac12$\textbf{BQP} to solve the Abelian HSP, but we have not yet found such an implementation.
 
To see why the phase estimation algorithm does not naturally carry over to the $\frac12$\textbf{BQP} setting, note that due to the random initial state in the $\frac12$\textbf{BQP} circuit, we cannot load the desired eigenvector into the phase estimation circuit. Instead, we must load a noisy state, which is no longer an eigenvector of the unitary whose phase we are trying to estimate.


\section{Separating $\frac{1}{2}$\textbf{BQP} and the polynomial Hierarchy}
\label{subsec:forrelation}

In this Section we show that $\frac12$\textbf{BQP} can solve \textsc{Forrelation} and the related Raz-Tal problem \cite{raz-tal}, implying an oracle separation with \textbf{PH}.

\textsc{Forrelation} is defined in \cite{aaronson2014forrelation} and provides a query complexity separation between classical and quantum computation. This is given by the following fact.
\begin{fact}
    Any classical randomized algorithm for \textsc{Forrelation} must make $\Omega(\frac{\sqrt{2^n}}{n})$ queries.
\end{fact}

We now define $k$-\textsc{Forrelation}. \textsc{Forrelation} is the same as $2$-\textsc{Forrelation}.

\begin{definition}[$k$-\textsc{Forrelation,aaronson2014forrelation}]
    Given phase oracle access to $k$ boolean functions $f_1,f_2,\dots,f_k$ where $f_i: \{0,1\}^n \rightarrow \{-1,1\}$, output whether $|\Phi|$ is $\leq \frac{1}{100}$ or $\Phi \geq \frac{3}{5}$, promised that one of these is the case, where 
    \[\Phi = \frac{1}{\sqrt{2^{(k+1)n}}} \sum_{x_1,x_2,\dots,x_k} f_1(x_1)(-1)^{x_1 \cdot x_2}f_2(x_2)(-1)^{x_2 \cdot x_3}\cdots (-1)^{x_{k-1}\cdot x_k}f_k(x_k)\]
\end{definition}

The \textbf{BQP} circuit for $k$-\textsc{Forrelation} is as follows:
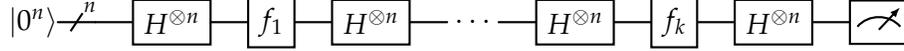
\begin{figure}[H]
    \centering
    \begin{quantikz}
        \ket{0^n} & \qwbundle{n} & \gate{H^{\otimes n}} & \gate{f_1} & \gate{H^{\otimes n}} & \ \cdots \ & \gate{H^{\otimes n}} & \gate{f_k} & \gate{H^{\otimes n}} & \meter{} 
    \end{quantikz}
\end{figure}
It is easy to show that the amplitude corresponding to $\ket{0^n}$ is exactly $\Phi$. Thus, a \textbf{BQP} machine can use the Hadamard test to accept or reject with success probability $\geq \frac{2}{3}$ using only $k$ quantum queries. 

\subsection{Placing 2-Forrelation in $\frac12$\textbf{BQP}}
\label{subsubsec:2-forrelation}
We now show that $\frac12$\textbf{BQP} solves $2$-\textsc{Forrelation}.

\begin{theorem}
    \label{thm:2-forrelation}
    $\frac12$\textbf{BQP} can solve $2$-\textsc{Forrelation} using $O(1)$ quantum queries.
\end{theorem}

\begin{proof}
    Consider the following $\frac12$\textbf{BQP} circuit $C$:
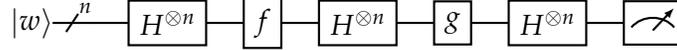
    \begin{figure}[H]
        \centering
        \begin{quantikz}
            \ket{w} & \qwbundle{n} & \gate{H^{\otimes n}} & \gate{f} & \gate{H^{\otimes n}} & \gate{g} & \gate{H^{\otimes n}} & \meter{} 
        \end{quantikz}
    \end{figure}
    Let $R = (-1)^{w\cdot z}$ be a random variable, where we measure $z$ and learn $w$. Then 
    \begin{align}
    \begin{split}
        \E[R] &= \sum_{w,z} \Pr[w,z] \cdot (-1)^{w \cdot z} \\
        &= \sum_{w,z} \Pr[w] \cdot \Pr[z | w] \cdot (-1)^{w \cdot z} \\
        &= \sum_w \frac{1}{2^n} \sum_z (-1)^{w \cdot z} \Pr[z | w]. 
    \end{split}
    \end{align}
    Now we analyze $\Pr[z|w]$ for fixed $z$ and $w$. Note that we use $+$ to denote bitwise addition mod $2$ and $\cdot$ for the dot product over $\mathbb{F}_2^n$. We also take $N = 2^n$. Let $\alpha_z = \bra{z}C\ket{w}$. Then 
    \begin{align}
    \begin{split}
        \alpha_z &= \frac{1}{N^{3/2}} \sum_{x,y} (-1)^{w\cdot x + x \cdot y + y \cdot z} f(x)g(y) \\
        &= \frac{1}{N^{1/2}} \sum_x (-1)^{w\cdot x} f(x) \frac{1}{2^n}\sum_y (-1)^{y\cdot (x + z)} g(y) \\
        &= \frac{1}{N^{1/2}} \sum_x (-1)^{w \cdot x} f(x) \widehat{g}_{x+z}
    \end{split}
    \end{align}
    and
    \begin{align}
    \begin{split}
        \Pr[z|w] = |\alpha_z|^2 &= \Biggl( \frac{1}{N^{1/2}} \sum_x (-1)^{w \cdot x} f(x) \widehat{g}_{x+z} \Biggr)^2 \\
        &= \frac{1}{2^n} \sum_{x,y} (-1)^{w \cdot (x + y)} f(x)f(y)\widehat{g}_{x+z}\widehat{g}_{y+z}.
    \end{split}
    \end{align}
    So
    \begin{align}
    \begin{split}
        \E[R] &= \sum_{w} \frac{1}{2^n} \sum_z (-1)^{w\cdot z} \frac{1}{2^n} \sum_{x,y} (-1)^{w \cdot (x + y)} f(x)f(y)\widehat{g}_{x+z}\widehat{g}_{y+z} \\
        &= \frac{1}{N^2} \sum_{x,y,z,w} f(x)f(y)\widehat{g}_{x+z}\widehat{g}_{y+z} (-1)^{w \cdot (x + y + z)} \\
        &= \frac{1}{N^2} \sum_{x,y,z} f(x)f(y)\widehat{g}_{x+z}\widehat{g}_{y+z} \sum_w (-1)^{w \cdot (x + y + z)} \\
        &= 
        \frac{1}{2^n} \sum_{x,y} f(x)f(y)\widehat{g}_{y}\widehat{g}_x.
    \end{split}
    \end{align}
    Now note that 
    \begin{align}
    \begin{split}
        \Phi &= \frac{1}{N^{3/2}} \sum_{x,y}f(x)g(y)(-1)^{x\cdot y} \\
        &= \frac{1}{N^{1/2}} \sum_x f(x) \frac{1}{2^n}\sum_y g(y)(-1)^{x\cdot y} \\
        &= \frac{1}{N^{1/2}} \sum_x f(x) \widehat{g}_x
    \end{split}
    \end{align}    
    so 
    \begin{align}
    \begin{split}
        |\Phi|^2 &= \Biggl( \frac{1}{2^n} \sum_x f(x) \widehat{g}_x \Biggr)^2 \\
        &= \frac{1}{2^n} \sum_{x,y} f(x)f(y)\widehat{g}_{y}\widehat{g}_x,
    \end{split}
    \end{align} 
    which is exactly $\E[R]$.
    Furthermore, since $R$ takes the value $1$ when $w\cdot z = 0$ and $-1$ when $w\cdot z = 1$, 
    \begin{align}
    \begin{split}
        \E[R] = |\Phi|^2 &= \Pr[w\cdot z = 0] - \Pr[w\cdot z = 1] \\ 
        &= \Pr[w\cdot z = 0] - (1 - \Pr[w\cdot z = 0]) = 2\Pr[w\cdot z = 0] - 1 \\
        \implies& \Pr[w\cdot z = 0] = \dfrac{1 + |\Phi|^2}{2}.
    \end{split}
    \end{align}
    Recall the promise that either $|\Phi| \leq \frac{1}{100}$ or $\Phi \geq \frac{3}{5}$. We run the $\frac12$\textbf{BQP} circuit $m$ times to obtain $R_1, R_2, \dots, R_m$ and accept if $\sum_{i=1}^m R_i \geq \frac{2m}{5}$.
     
    Suppose $|\Phi| \leq \frac{1}{100}$. Then by a loose Chernoff bound, setting $m=14$ gives \[\Pr[\sum_{i=1}^m R_i \geq \frac{2m}{5}] < \frac{1}{3}.\]
    Similarly if $\Phi \geq \frac{3}{5}$, for $m=14$,
    \[\Pr[\sum_{i=1}^m R_i \leq \frac{2m}{5}] < \frac{1}{3}.\]
    Thus, we can solve $2$-\textsc{Forrelation} in $\frac12$\textbf{BQP} using $O(1)$ quantum queries.
\end{proof}

\subsubsection{$\frac{1}{2}$\textbf{BQP} and the Forrelation Hierarchy}
\label{subsubsec:forrelation-hierarchy}

We have shown that $2$-\textsc{Forrelation} is in $\frac12$\textbf{BQP}. It turns out that $poly(n)$-\textsc{Forrelation} is promise-\textbf{BQP} complete \cite{aaronson2014forrelation}. Since the power of $k$-\textsc{Forrelation} seems to increase with the parameter $k$, we can consider a $k$-\textsc{Forrelation} ``hierarchy'' of increasing computational power. In particular since $2$-\textsc{Forrelation} is in $\frac12$\textbf{BQP}, if $\frac12$\textbf{BQP} $\subsetneq$ \textbf{BQP}, then there is some $t(n)$ for which $t(n)$-\textsc{Forrelation} is no longer in $\frac12$\textbf{BQP}. Finding this boundary will help understand the ways in which quantum computation over a random basis state is restricted relative to \textbf{BQP}.

Can $\frac12$\textbf{BQP} compute 3-\textsc{Forrelation}?
Here we give some evidence that $3$-\textsc{Forrelation} may not be in $\frac12$\textbf{BQP}.
Consider the following \textbf{BQP} circuit $C'$, which simulates the $\frac12$\textbf{BQP} procedure for $2$-\textsc{Forrelation}.

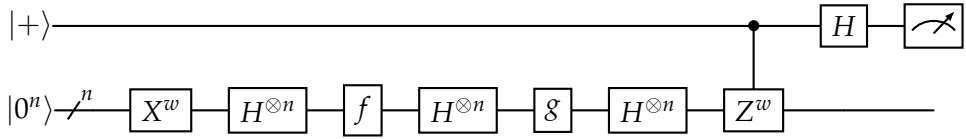
\begin{figure}[H]
    \centering
    \begin{quantikz}
        \ket{+} & & & & & & & & \ctrl{1} & \gate{H} & \meter{} \\
        \ket{0^n} & \qwbundle{n} & \gate{X^w} & \gate{H^{\otimes n}} & \gate{f} & \gate{H^{\otimes n}} & \gate{g} & \gate{H^{\otimes n}} & \gate{Z^w} & &
    \end{quantikz}
\end{figure}

Note that the controlled-$Z^w$ gate paired with the ancilla qubit simulates measuring the observable which corresponds to $(-1)^{wz}$. To see this, consider the following circuit:

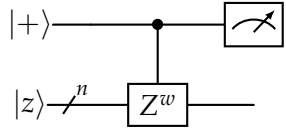
\begin{figure}[H]
    \centering
    \begin{quantikz}
        \ket{+} & & \ctrl{1} & \meter{} \\
        \ket{z} & \qwbundle{n} & \gate{Z^w} &
    \end{quantikz}
\end{figure}
The analysis is as follows:
\begin{align}
\begin{split}
    \ket{+}\ket{z} \xmapsto{c-Z^w}& \ket{0}\ket{z} + (-1)^{zw}\ket{1}\ket{z} \\
    =& 
    \begin{cases}
        \ket{+}\ket{z}, & zw = 0 \\
        \ket{-}\ket{z}, & zw = 1
    \end{cases}.
\end{split}
\end{align}
Now, using the Pauli commutation relations and the fact that $f, g$ are diagonal in the $Z$ basis, we have that $C' = $
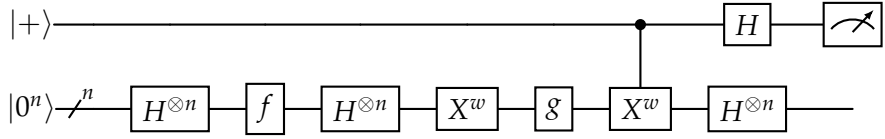
\begin{figure}[H]
    \centering
    \begin{quantikz}
        \ket{+} & & & & & & & \ctrl{1} & \gate{H} & \meter{} \\
        \ket{0^n} & \qwbundle{n} & \gate{H^{\otimes n}} & \gate{f} & \gate{H^{\otimes n}} & \gate{X^w} & \gate{g} & \gate{X^w} & \gate{H^{\otimes n}} &
    \end{quantikz}
\end{figure}

If we analyze the probability that the ancilla qubit is $0$, we obtain 
\[\Pr[0] = \frac{1}{2} + \frac{|\Phi|^2}{2}\]
as expected. 

The insight here is that we can think of $X^wgX^w$ as computing the function 
\[g^{\oplus w}: \{0,1\}^n \rightarrow \{-1,1\}\] given by
\[g^{\oplus w}(x) = g(x + w)\]

and then $k$-Forrelation can be computed as the expectation of $(-1)^{wz}$
where $w$ is uniformly random and $z$ is sampled with amplitude  
\[\alpha_z = \frac{1}{N^{(k+1)/2}}\sum_{x_1,\dots,x_k}(-1)^{x_1x_2 + \dots + x_{k-1}x_k} f_1(x_1)\dots f_{k-1}(x_{k-1}) f^{\oplus w}_{k}(x_k).\]

Importantly, there is not a nice commutation relation that allows us to pass $X^w$ through the second function $f_2$. Thus, this procedure does not seem to work in $\frac12$\textbf{BQP} for $k > 2$. The intuition is that if we think of starting with a random string $w$ instead of putting $X^w$ at the beginning of our computation, the $X^w$ gate can only commute through the first function $f_1$. While this analysis gives some insight as to why $\frac12$\textbf{BQP} may not be able to solve $3$-\textsc{Forrelation}, this is still an open question.

\subsection{Separating $\frac12$\textbf{BQP} and \textbf{PH}}
\label{subsec:ph-separation}
While \textsc{Forrelation} provides an almost tight query separation between quantum and classical computation, \cite{aaronson2014forrelation} only conjectured that there exists an oracle relative to which \textbf{BQP} $\not\subseteq$ \textbf{PH} and did not prove the claim. In 2018, \cite{raz-tal} proved this conjecture by giving a distribution based on \textsc{Forrelation} that hard to distinguish from the uniform distribution in \textbf{PH}, but can be distinguished in \textbf{BQP}. In this section we show that their problem is contained in $\frac12$\textbf{BQP}, extending their oracle separation to the $\frac12$\textbf{BQP} model. 

\subsubsection{The Raz-Tal Problem}
\label{subsubsec:raz-tal}
We now describe the Raz-Tal problem defined in \cite{raz-tal}.
Let $\varepsilon = \frac{1}{24n}$, and define a distribution $\mathcal{G}$ over $\mathbb{R}^{2^n} \times \mathbb{R}^{2^n}$ as follows: Sample $x_1,x_2,\dots,x_{2^n}$ independently from $\mathcal{N}(0,\varepsilon)$. Let $y = H_{2^n} x$ where $H_{2^n}$ is the Hadamard transform. $\mathcal{G}$ outputs $z = (x,y)$.

Now define the distribution $\mathcal{D}$ as follows: Sample $z$ from $\mathcal{G}$ but truncating each $z_i$ to the interval $[-1,1]$, and sample $z_1',z_2,',\dots,z_{2^n}'$ where $z_i' = 1$ with probability $\frac{1 + z_i}{2}$ and $z_i' = -1$ with probability $\frac{1 - z_i}{2}$, where $z_i$ is the truncated sample. $\mathcal{D}$ outputs $z' \in \{-1,1\}^{2\cdot 2^n}$.
\begin{definition}[The Raz-Tal problem]
    Given oracle access to a distribution over $\{-1,1\}^{2\cdot 2^n}$ which is either drawn from $\mathcal{D}$ or the uniform distribution $\mathcal{U}$, decide which is the case with success probability $\geq \frac{2}{3}$.
\end{definition}
The following facts from \cite{raz-tal} and \cite{aaronson2014forrelation} will be useful.
\begin{fact}
    \label{fact:raz-tal-amplitudes}
    Let $x,y \in \{-1,1\}^{2^n}$, and let 
    \[\Phi(x,y) = \frac{1}{2^{3n/2}} \sum_{i,j \in [2^n]}(-1)^{i\cdot j}x_iy_j.\]
    Then \cite{raz-tal}
    \[\E_{(x,y)\sim \mathcal{D}}[\Phi(x,y)] \geq  \frac{\varepsilon}{2}\]
    and \cite{aaronson2014forrelation}
    \[\E_{(x,y)\sim \mathcal{U}}[\Phi(x,y)^2] = \frac{1}{2^n}.\]
\end{fact}
To obtain the desired separation, \cite{raz-tal} define a new distribution $\mathcal{D}_2$ over $\{-1,1\}^{2\cdot 2^{n\cdot m}}$ which encodes $m=\frac{32\ln(1/\delta)}{\varepsilon^2} = poly(n)$ parallel copies of the Raz-Tal problem, where $\delta = \frac{1}{2^{poly(n)}}$ is a success parameter.

\begin{fact}
    \label{fact:raz-tal-circuit-lower-bound}
    No boolean circuit of size $2^{poly(n)}$ and constant depth distinguishes $\mathcal{D}_2$ from the uniform distribution $U$.
\end{fact}

We now give the $\frac12$\textbf{BQP} algorithm for the Raz-Tal problem. It is almost identical to the $\frac12$\textbf{BQP} algorithm for $2$-\textsc{Forrelation}.

\begin{theorem}
    \label{thm:raz-tal}
    $\frac12$\textbf{BQP} solves the Raz-Tal problem
\end{theorem}

Consider the following circuit $C$, where $U_{z'}$ is the oracle for $z' = (x,y)$.
\begin{figure}[H]
    \centering
    \begin{quantikz}
        \ket{0} &&& \gate[2]{U_{z'}} & \gate{X} & \gate[2]{U_{z'}} && \\ 
        \ket{w} & \qwbundle{n} & \gate{H^{\otimes n}} && \gate{H^{\otimes n}} && \gate{H^{\otimes n}} & \meter{}
    \end{quantikz}
    \caption{The $\frac12$\textbf{BQP} circuit for the Raz-Tal problem}
\end{figure}
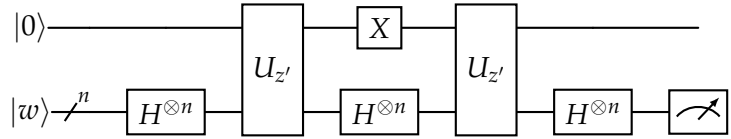
The first register encodes the most significant bit of the input to $U_{z'}$, so $C$ is equivalent to the following circuit, where $U_x, U_y$ encode $x$ and $y$, respectively. 
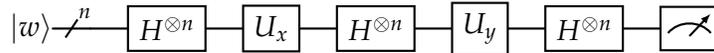
\begin{figure}[H]
    \centering
    \begin{quantikz}
        \ket{w} & \qwbundle{n} & \gate{H^{\otimes n}} & \gate{U_x} & \gate{H^{\otimes n}} & \gate{U_y} & \gate{H^{\otimes n}} & \meter{}
    \end{quantikz}
\end{figure}
Since $\Phi$ is symmetric in $x$ and $y$, we can assume $U_x$ is queried before $U_y$ without loss of generality. That is, it does not matter if the first qubit is $\ket{1}$ instead of $\ket{0}$ since the circuit will have the same output.
 
Now, note that this is exactly the circuit given for $2$-\textsc{Forrelation} \cref{subsec:forrelation}. The algorithm for the Raz-Tal problem is the same as before, where we output $R = (-1)^{w\cdot z}$ with $\E[R] = \Phi(x,y)^2$. Importantly, since the algorithm estimates $\Phi(x,y)^2$ rather than $\Phi(x,y)$, to place the Raz-Tal problem in $\frac12$\textbf{BQP}, we cannot immediately use \cref{fact:raz-tal-amplitudes}. We must show that $\E_{(x,y) \sim \mathcal{D}}[\Phi(x,y)^2]$ is sufficiently large to distinguish $D$ from $U$. This is captured by the following claim.
\begin{claim}
    \[\E_{(x,y) \sim \mathcal{D}}[\Phi(x,y)^2] \geq \frac{\varepsilon^2}{4}.\]
\end{claim}
\begin{proof}
    \begin{align}
    \begin{split}
        \E_{(x,y) \sim \mathcal{D}}[\Phi(x,y)^2] &= \Var_{(x,y) \sim \mathcal{D}}[\Phi(x,y)] + \E_{(x,y) \sim \mathcal{D}}[\Phi(x,y)]^2 \\
        &\geq \E_{(x,y) \sim \mathcal{D}}[\Phi(x,y)]^2 \\
        &\geq \frac{\varepsilon^2}{4}.
    \end{split}
    \end{align}
\end{proof}
Now, since $\frac12$\textbf{BQP} can distinguish $\mathcal{D}$ from $\mathcal{U}$ with advantage proportional to $\varepsilon^2$, $\frac12$\textbf{BQP} can distinguish $\mathcal{D}_2$ from the uniform distribution. The analysis is the same as given in \cite{raz-tal} but with an additional polynomial overhead in the number of parallel copies of the Raz-Tal problem encoded in the distribution $\mathcal{D}_2$. Thus, we have proven \cref{thm:raz-tal} and shown that there exists an oracle separation between $\frac12$\textbf{BQP} and \textbf{PH}.

\section{Separating $\frac12$\textbf{BQP} and \textbf{BQP}}
\label{subsec:bqp-separation}

So far we showed that $\frac{1}{2}$\textbf{BQP} captures many of the applications of quantum computing which we believe provide speedup over classical computation? Can $\frac{1}{2}$\textbf{BQP} simulate arbitrary quantum computations? We already saw that the usual approach we used to solve \textsc{Forrelation} in this model does not work for 3-\textsc{Forrelation}. Furthermore we do not know how to perform $\frac12$\textbf{BQP} computations with more than $O(\log n)$ clean qubits. Hence we would expect that $\frac{1}{2}$\textbf{BQP} is strictly weaker than \textbf{BQP}. In this section we provide more rigorous justifications for this statement. 

\cite{Knill_1998} gave a unitary oracle separation between \textbf{BQP} and \textbf{DQC}1 based on distinguishing two unitaries which have similar traces. This setup translates directly to a unitary oracle separation of \textbf{BQP} and $\frac12$\textbf{BQP}. We can also extend this setup to a classical oracle separation based on a ``lifting'' of Simon's problem given in \cite{chen2023complexity}. Before giving the theorem, we introduce the following lemmas, which will be useful in our analysis.

\begin{lemma}
    Let $A,B$ be $2^n \times 2^n$ unitary matrices with $\frac{1}{2^n}\Tr|A-B| \leq \varepsilon$.
    Then 
    \[
    \frac{1}{2^n}\Re\left(\Tr(A^\dagger B)\right) \geq 1 - \varepsilon.
    \]
\end{lemma}
\begin{proof}
First, by H{\"o}lder's inequality, 
    \begin{align}
    \begin{split}
        \frac{1}{2^n} |\Tr((A - B)^\dagger (A - B))| &\leq \frac{1}{2^n}\norm{A-B}_1 \norm{(A-B)^\dagger}_\infty \\
        &= \frac{1}{2^n}\Tr|A-B| \cdot \norm{A-B}_\infty \\
        &\leq \frac{1}{2^n}\Tr|A-B| \cdot (\norm{A}_\infty + \norm{B}_\infty) \\
        &\leq 2\varepsilon.
    \end{split}
    \end{align}
    Now, 
    \begin{align}
    \begin{split}
        \frac{1}{2^n} |\Tr((A - B)^\dagger (A - B))|
        &= \frac{1}{2^n} |\Tr(2I - (A^\dagger B + B^\dagger A))| \\
        &= 2 - \frac{2}{2^n}\Re(\Tr(A^\dagger B)) \leq 2\varepsilon
    \end{split}
    \end{align}
    so
    \begin{align}
        \frac{1}{2^n}\Re\left(\Tr(A^\dagger B)\right) \geq 1 - \varepsilon.
    \end{align}
\end{proof}

\begin{lemma}
    \label{lemma:tr}
    Let $A, B$ be $2^n \times 2^n$ unitary matrices and $0 \leq \varepsilon_1,\varepsilon_2 < 1$ such that $\frac{1}{2^n}|\Tr(A)| \geq 1 - \varepsilon_1$ and $\frac{1}{2^n} \Re(\Tr(B)) \geq 1 - \varepsilon_2$.
    Then
    \[ \frac{1}{2^n} |\Tr(AB)| \geq 1 - \varepsilon_1 - \sqrt{2 \varepsilon_2}. \]
    Furthermore, if $B$ is diagonal with eigenvalues $\pm 1$, then 
    \[ \frac{1}{2^n} |\Tr(AB)| \geq 1 - \varepsilon_1 - 2\varepsilon_2. \]
\end{lemma}
\begin{proof}
    \begin{align}
    \begin{split}
        \frac{1}{2^n} |\Tr(AB)| &= \frac{1}{2^n} |\Tr(A((B-I) + I))| \\
        &= \frac{1}{2^n} |\Tr(A) + \Tr(A(B-I))| \\
        &\geq \frac{1}{2^n}|\Tr(A)| - \frac{1}{2^n}|\Tr(A(B-I))| \\
        &\geq 1 - \varepsilon_1 - \frac{1}{2^n}|\Tr(A(B-I))|.
    \end{split}
    \end{align}
    By H{\"o}lder's inequality,
    \begin{align}
    \begin{split}
        \frac{1}{2^n}|\Tr(A(B-I))| &\leq \frac{1}{2^n} \norm{A^\dagger}_\infty \norm{B-I}_1 = \frac{1}{2^n}\norm{B-I}_1.
    \end{split}
    \end{align}
    Now let $\lambda_i(M)$ denote the $i$-th eigenvalue of a matrix $M$. Then

    \begin{align}
    \begin{split}
        \frac{1}{2^n}\norm{B-I}_1 &= \frac{1}{2^n}\sum_{i=1}^{2^n}|\lambda_i(B-I)| \\
        &= \frac{1}{2^n}\sum_{i=1}^{2^n}|1 - \lambda_i(B)| \\
        &= \frac{1}{2^n}\sum_{i=1}^{2^n}\sqrt{2 - 2\Re(\lambda_i(B))}.
    \end{split}
    \end{align}
    Subject to the condition $\frac{1}{2^n}\sum_{i=1}^{2^n} \Re(\lambda_i(B)) \geq 1 - \varepsilon_2$ (assumed in the statement) the above expression takes its maximum value when $\Re(\lambda_i(B)) = 1 - \varepsilon_2$ for all $i$. Thus,
    \begin{align}
    \begin{split}
        \frac{1}{2^n}\norm{B - I}_1 &\leq \frac{1}{2^n}\sum_{i=1}^{2^n}\sqrt{2 - 2(1 - \varepsilon_2)}
        = \frac{1}{2^n}\sum_{i=1}^{2^n}\sqrt{2 \varepsilon_2} = \sqrt{2\varepsilon_2}.
    \end{split}
    \end{align}
    Combining (11), (12), and (16) yields the desired result,
    \begin{align}
        \frac{1}{2^n}|\Tr(AB)| \geq 1 - \varepsilon_1 - \sqrt{2 \varepsilon_2}.
    \end{align}
    In the case where $B$ is diagonal with eigenvalues $\pm 1$, 
    \begin{align}
        \frac{1}{2^n} \norm{B - I}_1 = \frac{2}{2^n}\lfloor 2^n\varepsilon_2 \rfloor \leq 1 - 2\varepsilon_2,
    \end{align}
    so 
    \begin{align}
        \frac{1}{2^n}|\Tr(AB)| \geq 1 - \varepsilon_1 - 2 \varepsilon_2.
    \end{align}
\end{proof}

\begin{definition}[Distinguishing with advantage \cite{raz-tal}]
    Let $\mathcal{D}_1$ and $\mathcal{D}_2$ be probability distributions. We say that an algorithm $A$ distinguishes between $\mathcal{D}_1$ and $\mathcal{D}_2$ with \textit{advantage} $\delta$ if 
    \begin{align*}
        \Bigl| \Pr_{x \sim \mathcal{D}_1}[A \textrm{ accepts } x] - \Pr_{x \sim \mathcal{D}_2}[A \textrm{ accepts } x] \Bigr| \geq \delta.
    \end{align*}
    Furthermore, we say a model of computation distinguishes between unitaries $U$ and $U'$ with advantage $\delta$ if there exists an algorithm $A$ with access to oracle $O$ such that 
    \begin{align*}
        \Bigl| \Pr[A \textrm{ accepts} \mid O = U] - \Pr[A \textrm{ accepts} \mid O = U'] \Bigr| \geq \delta.
    \end{align*}
\end{definition}

We now give the theorem.

\begin{theorem}
    \label{thm:bqp-separation}
    Let $U$ and $U'$ be black-box oracles implementing $2^m \times 2^m$ unitaries such that 
    \[\frac{1}{2^m}\Tr|U - U'| \leq \varepsilon.\]
    Then $\frac12$\textbf{BQP} cannot distinguish between $U$ and $U'$ with advantage $\geq \delta$ using fewer than $\Omega(\frac{\delta^2}{\sqrt{\varepsilon}})$ queries.
\end{theorem}

\begin{proof}
    Consider a general $\frac12$\textbf{BQP} circuit on two halves of the EPR state, which performs $r$ rounds of queries to a unitary oracle $O$, and then measures according to an arbitrary POVM. Note that in the most general setting, we can use $n = m + l$ qubits in each half of our EPR state, where the oracle $O$ acts on $m$ qubits and $l$ qubits are used as workspace. However, we can view this as the situation in which the oracle is given by $O' = O \otimes I_{2^l}$ acting on $n$ qubits, in which case $\frac{1}{2^n}\Tr|(U - U') \otimes I_{2^l}| = \frac{2^l}{2^{m+l}}\Tr|U - U'| \leq \varepsilon$. Thus, it is sufficient to analyze the situation without additional workspace qubits.
    
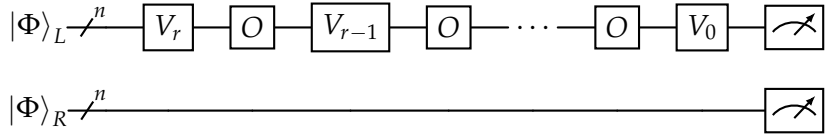
\begin{figure}[H]
        \centering
        \begin{quantikz}
            \ket{\Phi}_L & \qwbundle{n} & \gate{V_r} & \gate{O} & \gate{V_{r-1}} & \gate{O} & \; \dots \; & \gate{O} & \gate{V_0} & \meter{} \\
            \ket{\Phi}_R & \qwbundle{n} &&&&&&&& \meter{} \\
        \end{quantikz}
    \end{figure}
    Here, 
    \[\ket{\Phi}_{LR} = \frac{1}{\sqrt{2^n}} \sum_{x \in \{0,1\}^n} \ket{x}\ket{x}.\]
    We show that the resulting output states corresponding to $O = U$ and $O = U'$ have high fidelity, and therefore the probability of distinguishing these states is small for every choice of POVM.
     
    Let $\ket{\psi_U}$ and $\ket{\psi_{U'}}$ be the result of running the given circuit where $O = U$ and $O = U'$, respectively. Then the fidelity between $\ket{\psi_{U}}$ and $\ket{\psi_{U'}}$ after $r$ rounds is given by
    \begin{align}
    \begin{split}
        F_r(\ket{\psi_{U}},\ket{\psi_{U'}}) &= \frac{1}{2^n} \Bigl| \sum_{x,y \in \{0,1\}^n} \bra{x} V_r^\dagger U^\dagger V_{r-1}^\dagger \dots V_0^\dagger V_0 \dots V_{r-1} U' V_r \ket{y} \otimes \braket{x}{y} \Bigr| \\
        &= \frac{1}{2^n} \Bigl| \Tr(V_r^\dagger U^\dagger V_{r-1}^\dagger \dots V_0^\dagger V_0 \dots V_{r-1} U' V_r) \Bigr|.
    \end{split}
    \end{align}
    We now show by induction on the number of queries $r$ that $F_r(\ket{\psi_{U}},\ket{\psi_{U'}}) \geq 1 - O(r \sqrt{\varepsilon})$. For one query,
    \begin{align}
    \begin{split}
        F_1(\ket{\psi_{U}},\ket{\psi_{U'}}) &= \Bigl| \frac{1}{2^n} \Tr(V_1^\dagger U^\dagger V_0^\dagger V_0 U' V_1) \Bigr| \\
        &= \frac{1}{2^n}|\Tr(U^\dagger U')| \\
        &\geq \frac{1}{2^n}\Re(\Tr(U^\dagger U')) \geq 1 - \varepsilon \geq 1 - \sqrt{\varepsilon}.
    \end{split}
    \end{align}
    Now, suppose that
    \begin{align}
        F_{r-1}(\ket{\psi_{U}},\ket{\psi_{U'}}) = \frac{1}{2^n} \Bigl| \Tr(V_{r-1}^\dagger U^\dagger V_{r-2}^\dagger \dots V_0^\dagger V_0 \dots V_{r-2} U' V_{r-1}) \Bigr| \geq 1 - O(r\sqrt{\varepsilon}).
    \end{align}
    Then
    \begin{align}
    \begin{split}
        F_r(\ket{\psi_{U}},\ket{\psi_{U'}}) &= \frac{1}{2^n} \Bigl| \Tr(V_r^\dagger U^\dagger V_{r-1}^\dagger \dots V_0^\dagger V_0 \dots V_{r-1} U' V_r) \Bigr| \\
        &= \frac{1}{2^n} \Bigl| \Tr(U' U^\dagger V_{r-1}^\dagger \dots V_0^\dagger V_0 \dots V_{r-1}) \Bigr| \\
        &\geq 1 - O(r\sqrt{\varepsilon}) - \sqrt{2\varepsilon}  = 1 - O(r\sqrt{\varepsilon}).
    \end{split}
    \end{align}
    We now use the relationship between trace distance and fidelity to prove the theorem as stated.
    \begin{align}
    \begin{split}
        D_r(\ket{\psi_U}, \ket{\psi_{U'}}) &= \sqrt{1 - F_r(\ket{\psi_U}, \ket{\psi_{U'}})^2} \\
        &\leq O(r^{1/2} \varepsilon^{1/4}) = Cr^{1/2} \varepsilon^{1/4}
    \end{split}
    \end{align}
    for some constant $C$.
    Thus, if $\frac12$\textbf{BQP} can distinguish $U$ and $U'$ with advantage $\delta$, then $D_r \geq \delta$ and therefore
    \begin{align}
        r \geq \frac{\delta^2}{C\sqrt{\varepsilon}} = \Omega(\frac{\delta^2}{\sqrt{\varepsilon}}).
    \end{align}
\end{proof}

As a direct application of the theorem, $\frac12$\textbf{BQP} cannot distinguish between $U,U'$ using a sub-exponential number of queries when $\frac{1}{2^n}\Tr|U - U'| \leq O(\frac{1}{2^n})$.
 
This result allows us to give an explicit classical oracle separation between $\frac12$\textbf{BQP} and \textbf{BQP} using a construction from \cite{chen2023complexity} based on Simon's problem. The construction is as follows:
 
Let $f: \{0,1\}^n \rightarrow \{0,1\}^n$ be a function satisfying the conditions for Simon's problem, and let $\widetilde{f}: \{0,1\}^{2n} \rightarrow \{0,1\}^n$ given by 
\begin{align*}
    \widetilde{f}(x) =
    \begin{cases}
        f(x_1,\dots, x_n) & \textrm{if } x_{n+1},\dots,x_{2n} = 0 \\
        0 & \textrm{otherwise}
    \end{cases}.
\end{align*}
Here, the oracle $O_{\widetilde{f}}$ acts as
\[O_{\widetilde{f}} (\ket{x}\ket{y}) = \ket{x} | y \oplus \widetilde{f}(x) \rangle.\]
It is not hard to see that for two functions $f$ and $g$ with different hidden strings, \[ \frac{1}{2^{2n}} \Tr|O_{\widetilde{f}} - O_{\widetilde{g}}| \leq O(\frac{1}{2^n}).\] Thus, $\frac12$\textbf{BQP} cannot distinguish $f$ from $g$ and, therefore, cannot find the hidden string using a polynomial number of queries.

\subsection{$\frac{1}{2}$\textbf{BQP} lower bounds on unstructured search}
\label{subsec:grover}
We have seen that $\frac{1}{2}$\textbf{BQP} can solve many of the problems that exhibit quantum speedup over classical computations. It is reasonable to ask whether $\frac{1}{2}$\textbf{BQP} can achieve the $O(\sqrt{2^n})$ speedup for unstructured search given by Grover's algorithm \cite{grover1996fast}. It turns out the answer to this question is no. In particular, by a slight modification to the proof of \cref{thm:bqp-separation} we obtain an $\Omega(2^n)$ lower bound for unstructured search in the oracle setting. 

\begin{definition}[Unstructured search]
    Let $f: \{0,1\}^{n} \rightarrow \{0,1\}$ with the promise that there exists $\omega \in \{0,1\}^n$ such that
    \begin{align*}
        & f(x) = 1 \iff x = \omega.
    \end{align*}
    Given phase oracle $O_f$ which implements the function $f$, the task is to find $\omega$.
\end{definition}

\begin{theorem} [Lower bound for unstructured search]
    \label{thm:grover}
    Let $f,g$ be functions satisfying the conditions for unstructured search as defined above, and let $O_f, O_g$ be phase oracles implementing the functions $f,g $, respectively. Then $\frac{1}{2}$\textbf{BQP} cannot distinguish between $O_f$ and $O_g$ with advantage $\geq \delta$ using fewer than $\Omega(2^n \delta^2)$ queries.
\end{theorem}
\begin{proof}
    The proof is identical to \cref{thm:bqp-separation} except that since $O_f, O_g$ are diagonal with eigenvalues $\pm 1$, we use the bound corresponding to diagonal matrices with $\pm 1$ entries in \cref{lemma:tr}, which replaces the square root dependence on $\varepsilon$ with a linear dependence.
\end{proof}

\end{document}